\documentclass[11pt]{article}
\usepackage[compact]{titlesec}

\usepackage{times}
\usepackage{latexsym}
\usepackage{amsfonts,amsthm,amssymb}
\usepackage{amsmath}
\usepackage{euscript}
\usepackage{amstext}
\usepackage{graphicx}
\usepackage{color}
\usepackage{multicol}

\usepackage{url,hyperref}
\usepackage{verbatim}
\usepackage{xspace}
\usepackage[linesnumbered]{algorithm2e}
\usepackage{framed}

\usepackage{bbm}

\newtheorem{lemma}{Lemma}[section]
\newtheorem{theorem}[lemma]{Theorem}

\newtheorem{definition}[lemma]{Definition}
\newtheorem{corollary}[lemma]{Corollary}

\newtheorem{claim}[lemma]{Claim}

\newcommand{\optMS}{\textsc{OPT}{}_1}
\newcommand{\optMNS}{\textsc{OPT}{}_2}

        {\hspace*{\fill}$\Box$\par}

\newcommand{\opt}{\textsc{OPT}{}}

\newcommand{\One}{\mathbbm{1}}

\DeclareMathOperator*{\argmax}{arg\,max}

\def\Sigma{{\cal{S}}}
\def\N{{\cal{N}}}

\newcommand{\initOneLiners}{%
    \setlength{\itemsep}{0pt}
    \setlength{\parsep }{0pt}
    \setlength{\topsep }{0pt}
}

\newcommand{\defeq}{\stackrel{\text{def}}{=}}

\newcommand{\alg}{\mathsf{ALG}}
\newcommand{\pseudogreedy}{\mathsf{PseudoGreedy}}
\newcommand{\greedy}{\mathsf{Greedy}}

\newcommand{\E}{\mathbb{E}}

\usepackage{graphicx,xspace,bm,amsfonts, amsthm,amsmath}

\usepackage[numbers]{natbib}

\usepackage{url}

\newcommand{\UU}{{\mathbb U}}
\newcommand{\NN}{{\mathbb N}}
\newcommand{\RR}{{\mathbb R}}

\begin{document}
\title{Randomized Composable Core-sets for \\ Distributed Submodular Maximization
}
\author{Vahab Mirrokni \\
Google Research, \\ New York \\
mirrokni@google.com
\and
Morteza Zadimoghaddam \\
Google Research, \\
New York \\
zadim@google.com
\thanks{An extended abstract of this paper appeared at ACM Symposium on Theory of Computing (STOC)' 2015.}
}

\maketitle

\begin{abstract}
An effective technique for solving optimization problems over massive data sets is to partition the data into smaller pieces, solve the problem on each piece and compute a representative solution from it, and finally obtain a solution inside the union of the representative solutions for all pieces. This technique can be captured via the concept of {\em composable core-sets}, and has been recently applied to solve diversity maximization problems as well as several clustering problems~\cite{nips-BEL13,IMMM14,BBLM14}. However, for coverage and submodular maximization problems, impossibility bounds are known for this technique~\cite{IMMM14}.
 In this paper, we focus on efficient construction of a randomized variant of composable core-sets where the above idea is applied on a {\em random clustering} of the data. We employ this technique for the coverage, monotone and non-monotone submodular maximization problems. Our results significantly improve upon the hardness results for non-randomized core-sets, and imply improved results for submodular maximization in a distributed and streaming settings. The effectiveness of this technique has been confirmed empirically for several machine learning applications~\cite{nips13}, and our proof  provides a theoretical foundation to this idea. 
 
In summary, we show that a simple greedy algorithm results in a $1/3$-approximate randomized composable core-set for submodular maximization under a cardinality constraint. 
This is in contrast to a known $O({\log k\over \sqrt{k}})$ impossibility result for (non-randomized) composable core-set.
Our result also extends to non-monotone submodular functions, and leads to the first 2-round MapReduce-based constant-factor approximation algorithm with $O(n)$ total communication complexity for either monotone or non-monotone functions.  
Finally, using an improved analysis technique and a new algorithm $\pseudogreedy$, we present an improved $0.545$-approximation algorithm for monotone submodular maximization, which is in turn the first MapReduce-based algorithm beating factor $1/2$ in a constant number of rounds. 
\end{abstract}

\newpage

\section{Introduction}
An effective way of processing massive data is to first extract a compact representation of the data 
and then perform further processing only on the representation itself.
This approach significantly reduces the cost of processing, communicating and storing the data, as the representation size can be much smaller than the size of the original data set. Typically, the representation provides a smooth tradeoff between its size and the representation accuracy. 
Examples of this approach include techniques such as sampling, sketching, (composable) core-sets and mergeable summaries.
Among these techniques,  the concept of composable core-sets has been employed in several distributed optimization models such as nearest neighbor search~\cite{AAIMV13}, and the streaming and MapReduce models~\cite{nips13,nips-BEL13,IMMM14,karbasiKDD2014,BBLM14}.
Roughly speaking, the main idea behind this technique is as follows: First partition the data into smaller parts. Then compute a representative solution, referred to as a {\em core-set}, from each part. 
Finally, obtain a solution by solving the optimization problem over the union of core-sets for all parts. While this technique has been successfully applied to diversity maximization and clustering problems~\cite{nips-BEL13,IMMM14,BBLM14}, for coverage and submodular maximization problems, impossibility bounds are known for this technique~\cite{IMMM14}.

 In this paper, we focus on efficient construction of a randomized variant of composable core-sets where the above idea is applied on a {\em random clustering} of the data. We employ this technique for the coverage, monotone and non-monotone submodular problems. Our results significantly improve upon the hardness results for non-randomized core-sets, and imply improved results for submodular maximization in a distributed and streaming settings. The effectiveness of this technique has been confirmed empirically for several machine learning applications~\cite{nips13}, and our proof  provides a theoretical foundation to this idea.  Let us first define this concept, and then discuss its applications, and our results.
 
\subsection{Preliminaries} \label{sec:prelim}
Here, we discuss the formal problem definition, and the distributed model motivating it.

{\bf \noindent Submodular Functions.} 
We start by defining submodular functions~\footnote{ While the concepts in this paper can be applied to other set functions, we focus on maximizing submodular set functions.}.
Let  $\NN$ be a ground set of items with cardinality $n=\vert \NN\vert$.
Consider a set function $f: 2^{\NN} \rightarrow \RR^+\cup\{0\}$.
We say function $f$ is monotone if for  any two subsets $X \subseteq Y \subseteq \NN$, $f(X) \le f(Y)$. 
We say function $f$ is submodular if and only if for any two subsets $X \subseteq Y \subseteq \NN$, and an item $x \in \NN \setminus Y$,   we have the property of diminishing returns, i.e.,
$$ f(X \cup \{x\}) - f(X) \geq f(Y \cup \{x\}) - f(Y).$$ 
Given an integer size constraint $k$, we let $f_k$ be $$f_k(S)\defeq \max_{S'\subseteq S, \vert S'\vert \le k} f(S').$$
The submodular maximization problem with a cardinality constraint is as follows:
given a parameter $k$ and a value oracle access to a non-negative submodular function $f: 2^\NN\rightarrow \RR^+\cup\{0\}$, find a subset $S$ of cardinality at most  $k$ with the maximum value $f(S)$.
The most common algorithm for solving the above problem is algorithm $\greedy$ which is as follows: start from an empty set $S=\emptyset$, and in $k$ iterations, find an item $x$ with maximum marginal $f$ value for $S$ (i.e., $x=\argmax_{y\in \NN} f(S\cup\{y\}) - f(S)$) and add this item $x$ to $S$.  We refer to this algorithm as algorithm $\greedy$ and note that it is a $(1-{1\over e})$-approximation for monotone sumodular maximization problem with a cardinality constraint. 


{\bf \noindent Randomized Composable Core-sets.} In this paper, we assume that all $n$ items of $\NN$ do not fit on one machine, and 
we need to apply a distributed algorithm to solve submodular maximization problem.
To deal with this issue, 
we consider distributing items of $\NN$ into $m$ machines with indices $\{1, \ldots, m\}$, 
where each item goes to $C$ randomly chosen machines.
Let $\{T_1, T_2, \ldots, T_m\}$ be subsets of items going to machines $\{1, 2, \ldots, m\}$ respectively. In this case, we say that $\{T_1, T_2, \ldots, T_m\}$ is a {\em random clustering of $\NN$ with multiplicity $C$}, i.e., $\{T_1, T_2, \ldots, T_m\}$  is a family of subsets $T_i\subseteq \NN$, where each item of $\NN$ is assigned to $C$ randomly chosen subsets in this family.
Note that $T_i$'s are not necessarily disjoint subsets of items. 
Only the case of $C=1$ corresponds to a random partitioning of items into $m$ disjoint parts. This case
is the most natural way of applying this idea, and is studied  in Section~\ref{sec:rand-core-set}. As we see later, higher values of $C$  can help us achieve better approximation factors (See Section~\ref{sec:linear}). We are now ready to formally define randomized composable core-sets.

\begin{definition}
Consider an algorithm  $\alg$ that given any subset $T\subseteq \NN$ returns a subset ${\alg}(T) \subseteq T$ with size at most $k'$. Let $\{T_1, T_2, \ldots, T_m\}$ be a 
random clustering of $\NN$ to $m$ subsets  with multiplicity $C$.
We say that algorithm $\alg$ is an algorithm that implements 
{\em $\alpha$-approximate randomized composable core-set of size $k'$ with multiplicity $C$ for $f$ and cardinality constraint parameter $k$}, if, 
$$ \E \left [ f_k(\alg(T_1) \cup \ldots \cup \alg(T_m)) \right ] \ge {\alpha} \cdot \E\left [ f_k(T_1 \cup \ldots   \cup T_m)\right ],$$
where the expectation is taken over the random choice of $\{T_1, T_2, \ldots, T_m\}$. For brevity, instead of saying that $\alg$ implements a composable core-set, we say that $\alg$ is an $\alpha$-approximate randomized composable core-set. 
\end{definition}

For ease of notation, when it is clear from the context, we may drop the term composable, and refer to composable core-sets as core-sets. Throughout this paper, we discuss randomized composable core-sets for the submodular maximization problem with a cardinality constraint $k$.

{\bf \noindent Distributed Approximation Algorithm.} 
Note that we can use a randomized $\alpha$-approximate composable core-set algorithm $\alg$ to design the following simple distributed $(1-{1\over e})\alpha$-approximation algorithm for monotone submodular maximization:
\begin{enumerate}
\item  In the first phase, following the random clustering $\{T_1, \ldots, T_m\}$ defined above, allocate items in $\NN$ to $m$ machines, i.e., machine $i$ gets the subset $T_i$ of items. 
\item Each machine $i$ computes a randomized composable core-set $S_i\subseteq T_i$ of size $k'$, i.e., $S_i=\alg(T_i)$ for each $1\le i\le m$.
 \item In the second phase, first collect the union of all core-sets, $U=\cup_{1\le i\le m} S_i$, on one machine. Then apply a {\em post-processing} $(1-{1\over e})$-approximation algorithm (e.g., algorithm $\greedy$) to compute a solution $S$ to the submodular maximization problem over the set $U$. Output $S$.
\end{enumerate}

It follows from the definition of the $\alpha$-approximate randomized composable core-set that the above algorithm is a distributed $(1-{1\over e})\alpha$-approximation algorithm for submodular maximization problem. We refer to this two-phase algorithmic approach as {\em the distributed algorithm}, and the overall approximation factor of the distributed algorithm as the {\em distributed approximation factor}. For all our algorithms in this paper, in addition to presenting an algorithm that achieves an approximation factor $\alpha$ as a randomized composable core-set, we propose a post-processing algorithm for the second phase, and present an improved analysis that achieves much better than $(1-{1\over e})\alpha$-approximation as the distributed approximation factor. 

Note that the above algorithm can  be implemented in a distributed manner only if $k'$ is small enough such that $mk'$ items can be processed on one machine. In all our results the size of the composable core-set, $k'$, is a function of the cardinality constraint, $k$: In particular, in Section \ref{sec:rand-core-set}, we apply a composable core-set of size $k'=k$. In Section \ref{sec:linear}, we apply a composable core-set of size $k'<4k$, and as a result, achieve a better approximation factor.
We call a core-set, {\em a small-size core-set}, if its size $k'$ is less than $k$ (See Section~\ref{sec:small}). As we will see, the hardness results for small-size core-sets are much stronger than that of core-sets of size $k$ or larger.

{\bf \noindent Non-randomized Composable Core-sets.} The above definition for randomized composable core-sets is introduced in this paper. Prior work~\cite{IMMM14,BBLM14} define a non-randomized variant of composable core-sets where the above property holds for any (arbitrary) partitioning $\{T_1, T_2, \ldots, T_m\}$ of data into $m$ parts~\footnote{It is not hard to see that for non-randomized composable core-set, the multiplicity parameter $C$ is not relevant.}, i.e., an algorithm  $\alg$ as described above is a {\em $\alpha$-approximate (non-randomized) composable core-set of size $k'$ for $f$}, if for any cardinality constraint $k$, and any arbitrary partitioning $\{T_1, T_2, \ldots, T_m\}$ of the items into $m$ sets, we have 
$ f_k(\alg(T_1) \cup \ldots \cup \alg(T_m))  \ge {\alpha} \cdot f_k(T_1 \cup \ldots   \cup T_m)$.

\subsection{ Applications and Motivations} 
An $\alpha$-approximate randomized composable core-set of size $k'=O(k)$ for a problem can be applied in three types of applications~\cite{IMMM14}\footnote{These results assume  $k\le n^{1-\epsilon}$ for a constant $\epsilon$.}:
(i) in distributed computation~\cite{osdi-DG04}, where it implies an $\alpha$-approximation  in one or two rounds of MapReduces using the total communication complexity of $O(n)$,
(ii) in the random-order streaming model, where it implies an $\alpha$-approximation algorithm in one pass using sublinear memory,
(iii) in a class of approximate nearest neighbor search problems, where it implies an $\alpha$-approximation algorithm based on the locally sensitive hashing (under an assumption).
 Here, we discuss the application for the MapReduce and Streaming framework, and for details of the approximate nearest neighbor application, we refer to~\cite{IMMM14}.

We first show how to use a randomized composable core-set of size $O(k)$ to design a distributed algorithm in one or two rounds of MapReduces~\footnote{The straightforward way of applying the ideas will result in  two rounds
 of MapReduce. However, if we assume that the data is originally sharded randomly and each part is in a single shard, and the memory for each machine is more than the size
 of each shard, then it can be implemented via one round of MapReduce computation.} using linear total communication complexity: 
Let $m=\sqrt{n/k}$, and let $(T_1, \ldots, T_m)$ be a random partitioning where $T_i$ has $\sqrt{kn}$ items. In the distributed algorithm, we assume that the random partitioning is produced in one round of MapReduce where each of $m$ reducers receives $T_i$ as input, and produces a core-set $S_i$ for the next round. Alternatively, we may assume that the data (or the items) are distributed uniformly at random among machines, or similarity each of $m$ mappers receives $T_i$ as input, and produces a core-set $S_i$  for the reducer. In either case, the produced core-sets are passed to a single 
 reducer in the first or the second round. The total input to the reducer, i.e., the union of the core-sets, 
 is of size at most $mk'=O(k)\sqrt{n/k}=O(\sqrt{kn})$.
 The solution computed by the reducer for the union of the core-sets is, by definition, 
 a good approximation to the original problem. It is easy to see that the total communication complexity of this  algorithm is $O(n)$, and this computation can be performed in one or two rounds as formally defined in the MapReduce computation model~\cite{soda-KSV10}.  
 
 Next, we elaborate on the application for a streaming computation model:  In the random-order data stream model, a random sequence of $n$ data points needs to be processed ``on-the-fly'' while using only limited storage. An algorithm for a randomized composable core-set can be easily used to obtain an algorithm for this setting~\cite{GMMO,agarwal2004approximating}\footnote{The paper~\cite{GMMO} introduced this approach for the special case of $k$-median clustering. More general formulation of this method with other applications appeared in~\cite{agarwal2004approximating}.}. 
In particular, if a randomized composable core-set for a given problem has size $k$, we start by dividing the random stream of data into $\sqrt{n/k}$ blocks of size $s=\sqrt{nk}$. This way, each block will be a random subset of items.
The algorithm then proceeds block by block.
Each block is read and stored  in the main memory, its core-set is computed and stored, and the block is deleted. 
At the end, the algorithm solves the problem for the union of the core-sets. The whole algorithm takes only $O(\sqrt{kn})$ space.
The storage can be reduced further by utilizing more than one level of compression, at the cost of increasing the approximation factor.

Variants of the composable core-set technique have been applied for optimization under MapReduce framework~\cite{soda-KSV10,spaa-LMSV11,nips-BEL13,IMMM14,ANOY14,BBLM14}. However, none of these previous results formally study the difference between randomized and non-randomized variants and in most cases, they employ non-randomized composable core-sets. Indyk et al.~\cite{IMMM14} observed that the idea of non-randomized composable core-sets cannot be applied to the coverage maximization (or more generally submodular maximization) problems. 
In fact, all our hardness results also apply to a class of submodular maximization problems known as the maximum $k$-coverage problems, i.e.,
given a number $k$, and a family of subsets ${\cal A}\subset 2^X$, find a subfamily of $k$ subsets $A_1, \ldots, A_k$ whose union $\cup_{j=1}^k A_j$ is maximized.
Solving max $k$-coverage and submodular maximization in a distributed manner have  attracted a significant amount of research over the last few years~\cite{CKT10,CKW10,spaa-LMSV11,BST12,KMVV13,nips13}. Other than the importance of these problems, one reason for the popularity of this problem in this context is the fact that its approximation algorithm is algorithm  $\greedy$ which is naturally sequential and it is hard to parallelize or implement in a distributed manner. 


\subsection{Our Contributions} 
{
\begin{table*}
\begin{center}
\label{tab:summary-results1}
 \begin{tabular}[h]{|c|c|c|c|c|c|}
\hline 
Problem & Core-set Size & R/N & U/L & Core-set Approx. Factor & Distributed Approx. \\   \hline \hline
Mon. Submodular Max. ~\cite{IMMM14} & $\mbox{poly}(k)$ & N & U& $O(\frac{\log k}{\sqrt k})$& - \\ \hline
Mon. Submodular Max.* & $k$ & R & L& $1/3$ & 0.27  \\
Non-Mon. Submodular Max.*  & $k$ & R & L& $\max({m-1\over 3m},{1\over em})\ge 0.18$ &  $\max({1- {1 \over m}\over 2+e},{1\over em}) \ge 0.14$ \\
\hline \hline
Mon. Submodular Max.* & $O(k)$ & R & L& \small{$\iffalse{2-\sqrt{2}}\approx\fi 0.585 - O({1\over k})$} & $0.545 - O({1\over k})$ \\ 
Mon. Submodular Max. & $\mbox{poly}(k)$ & R & U& $1-{1\over e}$ & - \\ 
\hline
\hline
Mon. Submodular Max. & $k'<k$ & N & UL&$\Theta(\frac{k'}{k})$ & $\Theta(\frac{k'}{k})$ \\
\hline
Mon. Submodular Max.  & $k'<k$ & R & UL& $\Theta(\sqrt{\frac{k'}{k}})$&$\Theta(\sqrt{\frac{k'}{k}})$ \\
 \hline 
\end{tabular}
\end{center}
\caption{This table summarizes our results. 
In the column titled "R/N ", "R" corresponds
to the randomized core-set, and "N" corresponds to the non-randomized core-set notion.
In the column titled "U/L ", "U" corresponds
to an upper bound result, and "L" corresponds to a lower bound result. The last column corresponds 
to the distributed approximation factor. All these results except the first row are the new results of this paper. Previously, no constant-factor approximation has been proved for a randomized composable core-set for this problem. See Section~\ref{sec:relwork} for comparison to previous approximation algorithms. The rows with a star(*) are our most important results. 
}
\end{table*}
}

Our results are summarized in Table 1. As our first result, we prove that a family of efficient algorithms including a variant of algorithm $\greedy$  with a consistent tie-breaking rule leads to an almost $1/3$-approximate randomized composable core-set of size $k$ for any monotone submodular function and cardinality constraint $k$ with multiplicity of $1$  (see Section~\ref{sec:rand-core-set}). This is in contrast to a known $O({\log k\over \sqrt k})$ hardness result for any (non-randomized)  composable core-set~\cite{IMMM14}, and shows the advantage of using the randomization here.  
Furthermore, by constructing this randomized core-set and applying algorithm $\greedy$ afterwards, we show a $0.27$ distributed approximation factor for the monotone submodular maximization problem in {\em one or two rounds} of MapReduces with a {\em linear communication complexity}.  
 Previous results lead to algorithms with either much larger number of rounds of MapReduce~\cite{CKT10,BST12}, and/or larger communication complexity~\cite{KMVV13}. This improvement is important, since the number of rounds of MapReduce computation and communication complexity are the most important factors in determining the performance of a MapReduced-based algorithm~\cite{socc14}. The effectiveness of using this technique  has been confirmed empirically by Mirzasoleiman et al~\cite{nips13} who studied a similar algorithm on a subclass of submodular maximization problems. However, they only provide provable guarantees for a subclass of submodular functions satisfying a certain Lipchitz condition~\cite{nips13}. Our result not only works for monotone submodular functions, but also {\em extends to non-monotone (non-negative) submodular} functions, and leads to the first constant-round MapReduce-based constant-factor approximation algorithm for non-monotone submodular maximization (with $O(n)$ total communication complexity and approximation factor of $0.18$). It also leads to the first constant-factor approximation algorithm for non-monotone submodular maximization in a random-order streaming model in one pass with sublinear memory.
 
Our next goal is to improve the approximation factor of the above algorithm for monotone submodular functions. To this end, we first observe that one cannot achieve a better than the $1/2$ factor via core-sets of size $k$ using algorithm $\greedy$ or any algorithm in a family of local search algorithms.  In Section~\ref{sec:linear}, we show how to go beyond the $1/2$-approximation by applying core-sets of size higher than $k$ but still of size $O(k)$, and prove that  algorithm $\greedy$ with a consistent tie-breaking rule  provides a $0.585$-approximate randomized composable core-set of size $k'<4k$ for our problem. We then present algorithm $\pseudogreedy$ that can be applied as a post-processing step to design a distributed $0.545$-approximation  algorithm in one or two rounds of MapReduces, and with linear total communication complexity. For monotone submodular maximization, this result implies the first distributed approximation algorithm with approximation factor better than $1/2$ that  runs in a constant number of rounds. We achieve this approximation factor using one or two rounds of MapReduces and with the total communication complexity of $O(n)$. In addition, this result implies the first approximation algorithm beating the $1/2$ factor for the random-order streaming model with constant number of passes on the data and sublinear memory. To complement this result,  we first show that our analysis for algorithm $\greedy$ is tight.  Moreover, we show that it is information theoretically impossible to achieve an approximation factor better than $1-{1/e}$ using a core-set with size polynomial in $k$. 

Finally, we consider the construction of {\em small-size} core-sets, i.e., a core-set of size $k'<k$.  Studying such core-sets is important particularly for cases with large parameter $k$, e.g., $k= \Omega(n)$ or $k={n\over \log n}$~\footnote{For such large $k$, a core-set of size $k$ may not be as useful since outputting the whole core-set may be impossible. For example, in the formal MapReduce model~\cite{soda-KSV10}, outputting a core-set of size $k$ for $k=\Omega(n)$ is not feasible.}. 
For our problem, we first observe a hardness bound of $O({k'\over k})$ for  non-randomized core-sets. On the other hand, in Subsection~\ref{subset:small_size_submod}, we show an $\Omega(\sqrt{k'\over k})$-approximate randomized composable core-set for this problem, and accompany this result by a matching hardness  bound of $O(\sqrt{k'\over k})$ for randomized composable core-set. The hardness result is presented in Subsection~\ref{subsec:hardness}.

\subsection{Other Related Work.} \label{sec:relwork}

{\bf  Submodular Maximization in Streaming and MapReduce:}
Solving max $k$-coverage and submodular maximization in a distributed manner have  attracted a significant amount of research over the last few years~\cite{CKT10,CKW10,spaa-LMSV11,BST12,KMVV13,nips13,karbasiKDD2014}. From theoretical point of view, for the coverage maximization problem, Chierchetti et al.~\cite{CKT10} present a $(1-1/e)$-approximation algorithm in polylogarithmic number of MapReduce rounds, and Belloch et al~\cite{BST12} improved this result and achieved $\log^2 n$ number of rounds. Recently, Kumar et al.~\cite{KMVV13} present a $(1-1/e)$-approximation algorithm using a logarithmic number of rounds of MapReduces. They also derive $(1/2-{\epsilon})$-approximation  algorithm that runs in $O({1\over \delta})$ number of rounds of MapReduce (for a constant $\delta$), but this algorithm needs a $\log n$ blowup in the communication complexity. As observed in various empirical studies~\cite{socc14}, the communication complexity and the number of MapReduce rounds are important factors in determining the performance of a MapReduce-based algorithm and a $\log n$ blowup in the communication complexity can play a crucial role in applicability of the algorithm in practice.
Our algorithm on the other hand  runs only in (one or) two rounds, and can run on any number of machines as long as they can store the data, i.e, it needs $m$ machines each with memory proportional to ${1\over m}$ of the size of the input. One previous attempt to apply the idea of core-sets for submodular maximization is by Indyk et al.~\cite{IMMM14} who rule out the applicability of non-randomized core-sets by showing a hardness  bound of $O({\log k\over \sqrt k})$ for non-randomized core-sets. The most relevant previous attempt in applying randomized core-sets to  submodular maximization is by Mirzasoleiman et al~\cite{nips13}, where the authors study a class of algorithms similar to the ones discussed here, and show the effectiveness of applying algorithm $\greedy$ over a random partitioning empirically for several machine learning applications. The authors also prove theoretical guarantees for algorithm $\greedy$ for special classes of submodular functions satisfying a certain Lipschitz condition~\cite{nips13}. Here, on the other hand, we present guaranteed approximation results for all monotone and non-monotone submodular functions.
In fact, Ashwinkumar and Karbasi~\cite{AK14} observed that if one applies the greedy algorithm without a consistent tie-breaking rule, the approximation factor of the algorithm is not bounded for coverage functions. 
In this paper, we prove that a class of algorithms including greedy with a consistent tie-breaking rule provide a guaranteed $0.27$-approximation algorithm for all monotone submodular functions. We also show how to achieve an improved $0.54$-approximation using slightly larger core-sets of size $O(k)$. 
Finally, there is a recent paper in the streaming model~\cite{karbasiKDD2014} in which the authors present a streaming $1/2$-approximation algorithm with one pass and linear memory.  Our results lead to improved results for random-order streaming model for monotone and non-monotone submodular maximization.

{\bf Core-sets:}
The notion of core-sets has been introduced in~\cite{agarwal2004approximating}. 
In this paper, we use the term core-sets to refer to ``composable core-sets'' which was formally defined in a recent paper~\cite{IMMM14}. This notion has been also implicitly used in Section 5 
of Agarwal et al.~\cite{agarwal2004approximating} where the authors specify composability properties 
of $\epsilon$-kernels (a variant of core-sets). 
The notion of (composable) core-sets are also related to the concept of mergeable summaries 
that have been studied in the literature~\cite{agarwal2012mergeable}.
As discussed before, the idea of using core-sets has been applied either explicitly or 
implicitly in the streaming model~\cite{GMMO,agarwal2004approximating} and in the MapReduce 
framework~\cite{soda-KSV10,spaa-LMSV11,nips-BEL13,IMMM14,BBLM14}. 
Moreover, notions similar to randomized core-sets have been studied for random-order streaming models~\cite{KKS14}. Finally, the idea of random projection in performing algebraic projections can be viewed as a related topic~\cite{Sarlos06,CW13,FMSW10} but it does not discuss the concept of composability over a random partitioning.




\subsection{More Notation}\label{sec:notation}
We end this section by presenting notations and definitions used in the rest of the paper.
Let $\Delta(x,X)$ denote the marginal $f$ value of adding item $x$ to set $X$, i.e., $\Delta(x,X) \defeq f(X \cup \{x\}) - f(X)$.  
With this definition, the submodularity property is equivalent to saying that for any two subsets $X \subseteq Y \subseteq \NN$ and an item $x \in \NN \setminus Y$, we have $\Delta(x,X) \geq  \Delta(x,Y)$.
Consider the distributed algorithm described in Section~\ref{sec:prelim}. For any random clustering $\{T_1, \ldots, T_m\}$, we let $S_i$ ($1\le i\le m$)  be the output of algorithm $\alg$ on set $T_i$, i.e., $S_i = {\alg}(T_i)$. Also we let $\opt$ be the optimum solution ($\argmax_{|S| \leq k} f(S)$), and $\opt^{S}$ be the set of selected items of $\opt$ in $\cup_{i=1}^m S_i$, i.e., $\opt^{S}\defeq \opt \cap (\cup_{i=1}^m S_i)$. 
Throughout this paper, we analyze a class of algorithms, referred to as $\beta$-nice algorithms, with certain properties defined below:

\begin{definition}\label{def:nice}
Consider a submodular set function $f$. Let  $\alg$ be an algorithm that given any subset $T$ returns subset ${\alg}(T) \subseteq T$ with size at most $k'$. We say that this Algorithm $\alg$ is a {\em $\beta$-nice} algorithm for function $f$ and some parameter $\beta$ {\em iff} 
for any set $T$ and any item $x \in T \setminus {\alg}(T)$ (item $x$ is in set $T$ but is not selected in the output of algorithm $\alg$), then the following two properties hold: 
\begin{itemize}
\item Set ${\alg}(T \setminus \{x\})$ is equal to ${\alg}(T)$, i.e., intuitively the output of the algorithm should not depend on the items it does not select, and 
\item  $\Delta(x,{\alg}(T))$ is at most $\beta\frac{f({\alg}(T))}{k'}$. In other words, the marginal $f$ value of any not-selected item cannot be more than $\beta$ times the average contribution of selected items. 
\end{itemize}
\end{definition}

\section{Randomized Core-sets for Submodular Maximization}\label{sec:rand-core-set}
In this section, we show that a family of $\beta$-nice algorithms, introduced in Section~\ref{sec:notation}, leads to a constant-factor approximate randomized composable core-set, and a constant-factor distributed approximation algorithm for monotone and non-monotone submodular maximization problems with cardinality constraints.
Later, in Subsection~\ref{subsec:betaNiceAlgs}, we show that several efficient algorithms in the literature of submodular maximization  are $\beta$-nice for some $\beta \in [1, 1+\epsilon]$ (for $\epsilon=o(1)$) including some variant of  algorithm $\greedy$ with a consistent tie-breaking rule, and also an almost linear-time algorithm in \cite{AshwinVondrakSODA2014}.  Before stating the theorem, we emphasize that in this section, we apply a composable core-set of multiplicity 1 which corresponds to a {\em random partitioning} of items
into $m$ disjoint pieces.

\begin{theorem}\label{theorem:1/3_composable_coreset}
For any $\beta > 0$, any $\beta$-nice algorithm $\alg$ is a $\frac{1}{2+\beta}$-approximate  randomized composable core-set of multiplicity 1 and size $k$ for the monotone, and $\frac{1-\frac{1}{m}}{2+\beta}$-approximate for non-monotone submodular maximization problems with cardinality constraint $k$. 
\end{theorem} 
 
\begin{proof} 
We want to show that there exists a subset of $\cup_{i=1}^m S_i$ with size at most $k$, and at least an expected $f$ value of $\frac{f(\opt)}{2+\beta}$ for monotone $f$, and $\frac{(1-1/m)f(\opt)}{2+\beta}$ for non-monotone $f$. 
Toward this goal, we take the maximum of $\max_{1 \leq i \leq m} f(S_i)$ and $f(\opt^S)$  as a candidate solution. 
We define some notation to simplify the rest of the proof. Consider an arbitrary permutation $\pi$ on items of $\opt$, and for each item $x \in \opt$ let $\opt^x$ be the set of items in $\pi$ that appear before $x$.
We will first lower bound $f(\opt^S)$ using the submodularity property in Lemma~\ref{lem:linear_in_Deltas}. 

\begin{lemma}\label{lem:linear_in_Deltas}
For any set of selected items, 
$f(\opt^S)\ge f(\opt)-\sum_{x \in \opt \setminus (\cup_{i=1}^m S_i)} \Delta(x, \opt^x)$ for a monotone or non-monotone submodular function $f$.
\end{lemma}
\begin{proof}
First we note that $f(\opt) - f(\opt^S) = \sum_{x \in \opt \setminus \opt^S} \Delta(x, \opt^x \cup \opt^S)$. Using submodularity property, we know that $\Delta(x, \opt^x \cup \opt^S)$ is at most $\Delta(x,\opt^x)$ because $\opt^x$ is a subset of $\opt^x \cup \opt^S$. Therefore $f(\opt) - f(\opt \cap A) \leq \sum_{x \in \opt \setminus \opt^S} \Delta(x, \opt^x)$ which is equal to $\sum_{x \in \opt \setminus (\cup_{i=1}^m S_i)} \Delta(x, \opt^x)$ by definition of $\opt^S$. This concludes the proof.
\end{proof}

Lemma~\ref{lem:linear_in_Deltas} suggests that we should upper bound $\sum_{x \in \opt \setminus (\cup_{i=1}^m S_i)} \Delta(x, \opt^x)$ which is done in the next lemma. 

\begin{lemma}\label{lem:upperbound_not_selected}
The sum $\sum_{i=1}^m \sum_{x \in \opt \cap T_i \setminus S_i} \Delta(x, \opt^x)$ is at most 
 $\beta \left(\max_{1 \leq i \leq m} f(S_i)\right) +$ $\sum_{i=1}^m \sum_{x \in \opt \cap T_i \setminus S_i}$ $\left(\Delta(x, \opt^x) - \Delta(x, \opt^x \cup S_i)\right)$ for a monotone or non-monotone submodular function $f$.  
\end{lemma}
\begin{proof}
The sum $\sum_{i=1}^m \sum_{x \in \opt \cap T_i \setminus S_i} \Delta(x, \opt^x)$ can be written as:

$$
\sum_{i=1}^m \sum_{x \in \opt \cap T_i \setminus S_i} \left( \Delta(x, \opt^x \cup S_i) + \left(\Delta(x, \opt^x)- \Delta(x, \opt^x \cup S_i) \right) \right)
$$

The first term in the sum is upper bounded by $\beta\frac{f(S_i)}{k}$ using the second property of $\beta$-nice algorithms. To conclude the proof, we apply inequality $f(S_i) \leq \max_{1 \leq i' \leq m} f(S_{i'})$, and use the fact that  there are at most $k$ items in $\opt \setminus \opt^S = \cup_{i=1}^m (\opt \cap T_i \setminus S_i)$.
\end{proof}

At this stage of the analysis, we use
randomness of partition $\{T_i\}_{i=1}^m$ to upper bound the expected value of these differences in $\Delta$ values with the expected value of average of $f(S_i)$. This is stated in the following lemma:

\begin{lemma}\label{lem:randomness}
Assuming each item in $\opt$ is assigned to a $T_i$ uniformly at random (which is the case for all items not only members of $\opt$), we have that $\E[\sum_{i=1}^m \sum_{x \in \opt \cap T_i \setminus S_i} \Delta(x, \opt^x) - \Delta(x, \opt^x \cup S_i)]$ is at most $\frac{\E[\sum_{i=1}^m f(S_i)]}{m}$ for a monotone submodular function $f$, and at most  $\frac{\E[\sum_{i=1}^m f(S_i)]}{m} + \frac{f(\opt)}{m}$ for a non-monotone subdmodular function $f$. 
\end{lemma}

Before proving the lemmas, we observe that  putting these three lemmas together, we can finish the proof of the theorem. 
In particular, for a monotone or non-monotone submodular function $f$, we have that:

\begin{eqnarray*}
\E[f(\opt^S)] &\ge & 
f(\opt) - \beta \E[\max_{1 \leq i \leq m} f(S_i)] - \frac{\E[\sum_{i=1}^m f(S_i)]}{m} - \frac{f(\opt)}{m} \\
&\geq & (1-\frac{1}{m})f(\opt) -(1+\beta)\E[\max_{1 \leq i \leq m} f(S_i)].
\end{eqnarray*} 

 This immediately implies that 
 $\E[\max\{f(\opt^S), \max_{1 \leq i \leq m} f(S_i)\}]  \ge \frac{(1-\frac{1}{m})f(\opt)}{2+\beta},$ which concludes the proof for non-monotone functions. For monotone functions, we get the same proof except that we should exclude the $-\frac{f(\opt)}{m}$ term. 
To complete the proof, it is sufficient to prove the lemmas. 

{\noindent \bf Proof of Lemma~\ref{lem:randomness}.}
The main part of the proof is to show that the sum of the $\Delta$ differences  in the statement of the lemma is in expectation at most $\frac{1}{m}$ fraction of sum of $\Delta$ differences for a larger set of pairs $(i,x)$. In particular, we show that  


\begin{eqnarray*}\label{eq:two_sums}
\E[\sum_{i=1}^m \sum_{x \in \opt \cap T_i \setminus S_i} \Delta(x, \opt^x) - \Delta(x, \opt^x \cup S_i)] \leq
\frac{1}{m} \cdot  \E[\sum_{i=1}^m \sum_{x \in \opt} \Delta(x, \opt^x) - \Delta(x, \opt^x \cup S_i)].
\end{eqnarray*}

To simplify the rest of the proof, let $A$ be the left hand side of the above inequality, and $B$ be its right hand side. 
Also to simplify expressions $A$ and $B$, we introduce the following notation: For every item $x$ and set $T \subseteq \NN$, let $h(x, T)$ denote $\Delta(x, \opt^x) - \Delta(x, \opt^x \cup \alg(T))$. 
Also let $\One[x \notin \alg(T \cup \{x\})]$ be equal to one if $x$ is not in set $\alg(T \cup \{x\})$, and zero otherwise.
We note that $A$ and $B$ are both separable for different choices of item $x$ and set $T_i$, and can be rewritten formally using the new notation as follows: 
\begin{eqnarray*}
A &=& \sum_{i=1}^m \sum_{x \in \opt} \sum_{T \subseteq \NN \setminus \{x\}} Pr[T_i = T \cup \{x\}] \One[x \notin \alg(T \cup \{x\})] h(x, T \cup \{x\})\\ 
B &=& \sum_{i=1}^m \sum_{x \in \opt} \sum_{T \subseteq \NN \setminus \{x\}} \left(Pr[T_i = T \cup \{x\}] h(x, T \cup \{x\}) + Pr[T_i = T] h(x, T) \right) \\
&\geq& \sum_{i=1}^m \sum_{x \in \opt} \sum_{T \subseteq \NN \setminus \{x\}} \One[x \notin \alg(T \cup \{x\})] h(x, T \cup \{x\})\left(Pr[T_i = T \cup \{x\}] + Pr[T_i = T] \right) 
\end{eqnarray*}

{\noindent where} the inequality is implied by the following simple observations. Function $h$ is non-negative, so multiplying the sum by $\One[x \notin \alg(T \cup \{x\})]$ (which is either zero or one) can only decrease its value. We also replace one $h(x, T)$ with $h(x, T \cup \{x\})$ which does not change the value of the sum at all because when $\One[x \notin \alg(T \cup \{x\})]=1$ (its only non-zero value), $\alg(T \cup \{x\})$ is identical to $\alg(T)$ using the first property of nice algorithms, and thus $h(x, T)=h(x, T \cup \{x\})$ by definition. So the inequality holds. 

Now we can compare $A$, and $B$ as follows: For any set $T \subseteq \NN \setminus \{x\}$, we have $Pr[T_i = T]$ and $Pr[T_i = T \cup \{x\}]$ are equal to 
$\left(\frac{1}{m}\right)^{|T|}\left(1-\frac{1}{m}\right)^{|\NN|-|T|}$ and $\left(\frac{1}{m}\right)^{|T|+1}\left(1-\frac{1}{m}\right)^{|\NN|-|T|-1}$ respectively. As a result, the ratio $\frac{Pr[T_i = T \cup \{x\}]}{Pr[T_i = T \cup \{x\}] + Pr[T_i = T]}$ is equal to $\frac{1}{m}$ which shows that 
$A \leq B$.
 
To complete the proof for monotone $f$, it suffices to prove that  $B \leq \frac{\E[\sum_{i=1}^m f(S_i)]}{m}$, we note that for any $i$, $\sum_{x \in \opt}$ $\Delta(x, \opt^x \cup S_i)=f(\opt \cup S_i) - f(S_i),$ and  $\sum_{x \in \opt} \Delta(x, \opt^x)=f(\opt).$  For a monotone $f$, we have $f(\opt \cup S^i) \geq f(\opt)$ which shows that  $B \leq\frac{\E[\sum_{i=1}^m f(S_i)]}{m}$.

Although, we do not necessarily have $f(\opt \cup S_i) \geq f(\opt)$ for a non-monotone $f$, we show that $\sum_{i=1}^m f(\opt \cup S_i) \geq (m-1)f(\opt)$. For any $1 \leq i < m$, we have that $f\left(\opt \cup \left( \cup_{i'=1}^i S_{i'}\right)\right) + f(\opt \cup S_{i+1})$  is at least $f(\opt) + f\left(\opt \cup \left( \cup_{i'=1}^{i+1}+ S_{i'}\right)\right)$ by submodularity. Applying this inequality for each $1 \leq i < m$, and using non-negativity of $f$, we imply that $\sum_{i=1}^m f(\opt \cup S_i) \geq (m-1)f(\opt)$. Therefore $B$ is at most $\frac{\E[\sum_{i=1}^m f(S_i)]}{m} + \frac{(m-(m-1))f(\opt)}{m}$, and the proof is completed for non-monotone $f$ as well.\\
\vspace{-2mm}
\end{proof}



\subsection{Overall Distributed Approximation Factor of $\beta$-nice Algorithms}
In Theorem~\ref{theorem:1/3_composable_coreset}, we prove that if on each part (set $T_i$) of the partitioning, we run a $\beta$-nice algorithm $\alg$, the union of output sets of $\alg$ will contain a set of size at most $k$ that preserves at least $\frac{1}{2+\beta}$ fraction of value of optimum set. If we run algorithm $\greedy$ on the union of output sets $\cup_{i=1}^m S_i$, using the classic analysis of algorithm $\greedy$, we can easily claim that the overall value of the output set at the end is at least $\frac{1-1/e}{2+\beta}$ fraction of $f(\opt)$. In Theorem~\ref{thm:overall_performance_greedy}, 
we show an improved distributed approximation factor.


\begin{theorem}\label{thm:overall_performance_greedy}
Let $S$ be the output of algorithm $\greedy$ over $\cup_{i=1}^m S_i$, i.e. $S = \greedy(\cup_{i=1}^m S_i)$ for a monotone submodular function $f$. Also let $S$ be the output of non-monotone submodular maximization algorithm of Buchbinder et al. \cite{BuchbinderFeldmanNaorSchwartzSODA2014} on $\cup_{i=1}^m S_i$ when $f$ is a non-monotone submodular function. The expected value of $\max\{f(S), \max_{i=1}^m\{f(S_i)\}\}$ is at least $\frac{(1-1/e)f(\opt)}{1+(1-1/e)(1+\beta)}$ for monotone and $\frac{(1-\frac{1}{m})f(\opt)/e}{1+(1+\beta)/e}$ for non-monotone submodular $f$.  In particular, for $\beta= 1$, the distributed approximation factors are $\ge 0.27$, and $\ge 0.21 - \frac{1}{4m}$ for monotone and non-monotone $f$ respectively.
\end{theorem}
\begin{proof}
By applying lemmas \ref{lem:linear_in_Deltas}, \ref{lem:upperbound_not_selected}, and \ref{lem:randomness}, we have that: $f(\opt^S) \geq f(\opt) - \beta\max_{i=1}^m f(S_i) - \frac{\sum_{i=1}^m f(S_i)}{m} \geq f(\opt) - (1+\beta) \max_{i=1}^m f(S_i)$ for a monotone submodular function $f$. Using the classic analysis of $\greedy$ on submodular maximization in~\cite{Nemhauser_Wolsey_Fisher78}, one can prove that $f(S) \geq (1-1/e) f(\opt^S)$ when $f$ is monotone. 
By taking the expectation of the two sides of this inequality, and using Lemma~\ref{lem:randomness}, we have that: 

$$
\E[f(S)] \geq 
\E\left[(1-1/e) \left(f(\opt) - (1+\beta)\max_{i=1}^m f(S_i)\right)\right] 
$$

By letting $\rho :=\frac{\E[\max\{f(S), \max_{i=1}^m \{f(S_i)\}\}]}{f(\opt)}$,
 the above inequality implies that  $\rho\ge (1-1/e)(1-(1+\beta)\rho)$, and therefore $\rho \geq \frac{1-1/e}{1+(1-1/e)(1+\beta)}$. Consequently $\E[\max\{f(S), \max_{i=1}^m$ $\{f(S_i)\}\}]$ is at least $\frac{(1-1/e)f(\opt)}{1+(1-1/e)(1+\beta)}$ which proves the claim for any monotone $f$. 

 If $f$ is non-monotone, we get a weaker inequality $f(S) \geq \frac{f(\opt^S)}{e}$ by applying the algorithm of Buchbinder et al. \cite{BuchbinderFeldmanNaorSchwartzSODA2014}. 
 Using lemmas \ref{lem:linear_in_Deltas}, \ref{lem:upperbound_not_selected}, and \ref{lem:randomness}, we have that: $f(\opt^S) \geq f(\opt) - \beta\max_{i=1}^m f(S_i) - \frac{\sum_{i=1}^m f(S_i)}{m} - \frac{f(\opt)}{m} \geq (1-\frac{1}{m})f(\opt) - (1+\beta) \max_{i=1}^m f(S_i)$. 
 Similarly, we can claim that  $\rho \geq \frac{(1-\frac{1}{m})/e}{1+(1+\beta)/e}$. This implies the desired lower bound of $\frac{(1-\frac{1}{m})/e}{1+(1+\beta)/e}$ on $\rho$.
\end{proof}

{\bf \noindent Non-monotone submodular maximization}
In Theorem~\ref{theorem:1/3_composable_coreset}, we proved that for a non-monotone $f$, $\greedy$ returns a $(\frac{1}{3} - \frac{1}{3m})$-approximate randomized composable core-set. When we have a small number of machines, $m$, this approximation factor becomes small.  So we suggest the non-monotone submodular maximization algorithm of Buchbinder et al. \cite{BuchbinderFeldmanNaorSchwartzSODA2014} as an alternative. We note that since the items are partitioned randomly, $\frac{1}{m}$ fraction of optimum solution is sent to each machine in expectation, and the expected value $\E[f(\opt \cap S_i)]$ is at least $\frac{f(\opt)}{m}$ for any $1 \leq i \leq m$. Using the algorithm of Buchbinder et al. we get a $\frac{1}{em}$-approximate randomized core-set. We also note that the $0.21$-approximation guarantee for non-monotone $f$ in Theorem~\ref{thm:overall_performance_greedy} holds for large enough number of machines $m$. For small $m$, this $\frac{1}{em}$ approximation can be used as an alternative. 

\subsection{Examples of $\beta$-Nice Algorithms}\label{subsec:betaNiceAlgs}
In this section, we show that several existing algorithms for submodular maximization in the literature belong to the family of $\beta$-nice algorithms.  

{\bf \noindent Algorithm $\greedy$ with Consistent Tie-breaking:} 
First, we observe that algorithm $\greedy$ is $1$-nice if it has a consistent tie breaking rule:  while selecting among the items with the same marginal value, $\greedy$ can have a fixed strict total ordering ($\Pi$) of the items, and among the set of items with the maximum marginal value chooses the one highest rank in $\Pi$. The consistency of the tie breaking rule implies the first property of nice algorithms. To see the second property, first observe that  (i) $\greedy$ always adds an item with the maximum marginal $f$ value, and (ii) using submodularity of $f$, the marginal $f$ values are decreasing as more items are added to the selected items. Therefore, after $k$ iterations, the marginal value of adding any other item, is less than each of the $k$ marginal $f$ values we achieved while adding the first $k$ items. This implies the 2nd property, and concludes that $\greedy$ with a consistent tie-breaking rule is a $1$-nice algorithm. 

{\bf \noindent An almost linear-time $(1+\epsilon)$-nice  Algorithm:} Badanidiyuru and Vondrak \cite{AshwinVondrakSODA2014} present an almost linear-time $(1-\frac{1}{e}-\epsilon)$-approximation algorithm for monotone submodular maximization with a cardinality constraint. We observe that this algorithm is $(1+2\epsilon)$-nice. The algorithm is a relaxed version of $\greedy$ where in each iteration, it adds an item with almost maximum marginal value (with at least $1-\epsilon$ fraction of the maximum marginal). As a result, similar to the proof for $\greedy$, one can show that this linear-time algorithm $\frac{1}{1-\epsilon}$-nice and consequently $(1+2\epsilon)$-nice for $\epsilon \leq 0.5$.

\section{Hardness Results for Randomized Core-sets}\label{sec:hardness_results}
In Section~\ref{sec:rand-core-set}, we showed that a family of $\beta$-nice algorithms are $\frac{1}{2+\beta}$-approximate randomized core sets (e.g., $\frac{1}{3}$-approximate for algorithm $\greedy$). 
Here we show what kinds of randomized core-sets are not achievable.  
In particular, we prove, in Theorem~\ref{thm:one_half_barrier} that if we restrict our attention to core-sets of size $k$, algorithm $\greedy$ or any local search algorithm does not achieve an approximation factor better than $\frac{1}{2}$ even if each item is sent to multiple machines (up to multiplicity $C = o(\sqrt{m})$). This leads to the following question: does increasing the output size of core-sets, $k'$, help with the approximation factor? In other words, can we get a better than $1/2$ approximation factor if we allow the algorithm to select more than $k$ items on each machine?
To answer this question, we first prove, in Theorem~\ref{thm:OneMinusOneOverEHardness} that it is not possible to achieve a randomized composable core-set of size $k' = o(\frac{n}{Cm})$ with approximation better than $1 - \frac{1}{e}$ even when we allow for multiplicity $C = o(\sqrt{\frac{m}{k}})$. 
We then show in Section~\ref{sec:linear} that although it is not possible to beat the $1-\frac{1}{e}$ barrier, we can slightly increase  the output sizes, apply algorithm $\greedy$ to achieve an approximation factor $\approx 2-\sqrt{2} > \frac{1}{2}$ with a constant multiplicity.

Following we show a limitation on core-sets of size $k$. In particular, we introduce a family of instances for which algorithm $\greedy$ and any algorithm that returns a locally optimum solution of size at most $k$ do not achieve a better than $\frac{1}{2}+\epsilon$-approximate core-set for any $\epsilon > 0$. This lower-bound result applies to a coverage valuation (and therefore submodular) function $f$ and it holds even if we send each item to multiple machines. 

\begin{theorem}\label{thm:one_half_barrier}
For any $\epsilon > 0$, assuming each item is sent to at least one random machine (to set $T_i$ for a random $1 \leq i \leq n$), and at most $C \leq \sqrt{\frac{\epsilon m}{2}}$ random machines, and the number of items an algorithm is allowed to return is at most $k$, there exists a family of instances for which algorithm $\greedy$ and any other local search algorithm returns an at most $(\frac{1}{2} + \frac{1}{k} + \epsilon)$-approximate composable core-set.  
\end{theorem}    

\begin{proof}
Let $\NN$ be a subfamily of subsets of a universe ground set $\UU$, and let
$f:2^\NN\rightarrow R$ be defined as follows: for any set $S \subseteq \NN$, $f(S) = |\cup_{A \in S} A|$. 
Note that $f$ is a coverage function, and thus it is a monotone submodular function. 
Let the universe $\UU$ be $\{1, 2, \cdots, k^2+(k-1)^2\}$. 
The family $\NN$ of subsets of $\UU$ consists of two types of subsets: a) $k$ sets $\{A_i\}_{i=1}^k$, and b) $kL$ sets $B_{i,j}$ for any $1 \leq i \leq k$, and $1 \leq j \leq L$ for some large $L > m\ln\left(\frac{\epsilon}{2Ck}\right)$. 
Let the first set $A_1$ be the subset $\{1, 2, \cdots, k^2\}$. For each $2 \leq i \leq k$, we define $A_i$ to be the set $\{k^2+(k-1)(i-2)+1, k^2+(k-1)(i-2)+2, \cdots, k^2+(k-1)(i-1)\}$ with size $k-1$. 
Therefore the first type of sets $\{A_i\}_{i=1}^k$ form a partitioning of the universe, and therefore they are the optimum family of $k$ sets with the maximum $f$ value. 
For each $1 \leq i \leq k$, all sets $B_{i,j}$ are equal to $\{(i-1)k+1, (i-1)k+2, \cdots, ik\}$ with size $k$. So all type $B$ sets with the same $i$ value  (and different $j$ values) are identical.

We say a machine is a {\em good} machine if for each $1 \leq i \leq k$, it receives at least a set $B_{i,j}$ for some $1 \leq j \leq L$, and we call it is a {\em bad} machine otherwise. At first, we show that the output of algorithm $\greedy$ or any local search algorithm on a good machine that has not received set $A_1$ is exactly one set $B_{i,j}$ for each $1 \leq i \leq k$, and nothing else. In other words, these algorithms do not return any of the sets $A_2, A_3, \cdots, A_k$ unless they have a set $A_1$ as part of their input, or they are running on a bad machine. It is not hard to see that if $A_1$ is not part of the input each set $B_{i,j}$ has marginal value $k$ if no other set $B_{i,j'}$ for $j' \neq j$ has been selected before. On the other hand, the marginal value of each of the  sets $\{A_i\}_{i=2}^k$ is $k-1$. So $\greedy$ or any local search algorithm does not select any of the sets  $\{A_i\}_{i=2}^k$ unless some set $B_{i,j}$ has been selected for each $1 \leq i \leq k$. The fact that the output sizes are limited to $k$ implies that sets $\{A_i\}_{i=2}^k$ are not selected.

Now it suffices to prove that most machines are good, and most sets in $\{A_i\}_{i=2}^k$ are not sent to a machine that has set $A_1$ as well. To prove this, we show that each machine is good with probability at least $1-\frac{\epsilon}{2C}$. To see this, note that for each $i$, there are $L$ identical sets $B_{i,j}$, and the probability that a machine does not receive any of these $L$ copies is at most $(1-\frac{1}{m})^L \leq e^{-L/m} \leq \frac{\epsilon}{2Ck}$. So the probability that a machine is good is at least $(1-\frac{\epsilon}{2Ck})^k \geq 1- \frac{\epsilon}{2C}$. Since each set is sent to at most $C$ machines, for each $1 \leq i \leq k$, we know that set $A_i$ is sent to only good machines with probability at least $(1-\frac{\epsilon}{2C})^C \geq 1-\frac{\epsilon}{2}$. We also know that the probability that for each $2 \leq i \leq k$, the probability of set $A_i$ sharing a machine with set $A_1$ is at most $\frac{C^2}{m} \leq \frac{\epsilon}{2}$ since each set is sent to at most $C$ random machines. As a result, at most $\frac{\epsilon}{2}+\frac{\epsilon}{2}=\epsilon$ fraction of sets  $\{A_i\}_{i=2}^k$ are selected by at least one machine in expectation, and therefore in expectation the size of the union of all selected sets (not only the best $k$ of them) is at most $k^2+\epsilon (k-1)^2$ which is less than $\frac{1}{2} + \frac{1}{k} + \epsilon$ fraction of the value of the optimum $k$ sets ($\{A_i\}_{i=1}^k$).
\end{proof}

Following, we prove that it is not possible to achieve a better than $1 - \frac{1}{e} + \epsilon$ approximation factor for submodular maximization subject to a cardinality constraint even if 
each item is sent to at most $C$ machines where $C \leq \sqrt{\epsilon m}$, and
each machine is allowed to return $k' = \frac{\epsilon (1-\epsilon) n}{8Cm}$ items. We note that in this hardness result, the size of the output sets can be arbitrarily large in terms of $k$, i.e. for instance $k'$ could be $\Omega(2^k)$ for some values of $n, m, C$, and $k$.   This is an information theoretic hardness result that does not use any complexity theoretic assumption. In fact, the instance itself can be optimally solved on a single machine, but distributing the items among several machines makes it hard to preserve the optimum solution. Before presenting the hardness result, we state the following version of Chernoff bound  (which we use in the proof) as given on
page 267, Corollary $A.1.10$ and Theorem $A.1.13$ in \cite{AlonSpencerBook2000}:

\begin{lemma}\label{lem:ChernoffBound}
Suppose $X_1, X_2, \cdots, X_n$ are $0-1$ random variables such that $Pr[X_i=1] = p_i$, and let $\mu = \sum_{i=1}^n p_i$, and $X = \sum_{i=1}^n X_i$. Then for any $a > 0$

$$
Pr[X-\mu \geq a] \leq e^{a-(a+\mu)\ln(1+a/\mu)} 
$$

Moreover for any $a > 0$

$$
Pr[X-\mu \leq -a] \leq e^{-a^2/\mu} 
$$
\end{lemma} 

\begin{theorem}\label{thm:OneMinusOneOverEHardness}
For any $\epsilon > 0$, $k \geq \frac{8}{\epsilon}$ and $C \leq \sqrt{\frac{\epsilon m}{4k}}$, assuming each item is sent to at most $C$ machines randomly, and each machine can output at most $k' = \frac{\epsilon(1-\epsilon)}{8C} \times \frac{n}{m}$ items, there exists a family of instances for which no algorithm can guarantee a core-set of expected value more than $(1 - \frac{1}{e} + \epsilon)$ fraction of the optimum solution.  
\end{theorem}

\begin{proof}
Similar to the proof of Theorem\ref{thm:one_half_barrier}, 
we define $\NN$ to be a family of subsets of the ground set $\UU$. 
We assume that $|\UU| \gg |\NN|=n$, and also assume that $|\UU| = k\ell$ for some integer $\ell$ . 
For any set $S \subseteq \NN$, let $f(S) = |\cup_{A \in S} A|$ which is number of items in the ground set $\UU$ that are present in at least one of the sets of $S$. 
We define the $n$ items $\NN = \{A_1, A_2, \cdots, A_n\}$ as follows. Each item $A_i$ is a subset of $\UU$ with size $\ell$.
The first $k$ items $\{A_i\}_{i=1}^k$ form a random partitioning of the ground set $\UU$ (they are $k$ disjoint sets each with size $\ell$ to cover all $k\ell = |\UU|$ elements of $\UU$). So these $k$ items form an optimum solution with value $|\UU| = k \ell$. 
For each $k < i \leq n$, we define each item $A_i$ to be a random subset of $\UU$ with size $\ell$. 

Each item $A_i$ ($i \leq k$) of the optimum solution is a random subset of $\UU$ with size $\ell$, so in a machine $A_i$ cannot be distinguished from items not in the optimum solution unless some other item $A_j$ (for some $j \leq k$) is present in the machine.  
Therefore any machine that receives only one item of the optimum solution will return it with probability at most $\frac{k'}{(1-\epsilon)n/m} + e^{-\frac{\epsilon^2n/m}{2}} \leq \frac{\epsilon}{8C} + \frac{\epsilon}{8C}$. Because each machine returns at most $k'$ items, and with probability at least $1 - e^{-\frac{\epsilon^2n/m}{2}}$, it has at least $(1-\epsilon)\frac{n}{m}$ items. We note that since each item is sent to at least one machine, the average load is at least $\frac{n}{m}\geq \frac{2\ln(8C/\epsilon)}{\epsilon^2}$.   

So among the optimum solution items that are alone in their machines, at most $\frac{\epsilon}{4C}$ fraction of them will be selected. We note that each item is sent to at most $C$ machines, therefore in total $\frac{\epsilon k}{4}$ optimum items will be selected in this category in expectation. 

On the other hand, the total number of optimum items that share a machine with some other optimum item is at most $\frac{(Ck)^2}{m} \leq \frac{\epsilon k}{4}$ (an upper bound on the expected number of collisions).   We conclude that in total at most $\frac{\epsilon k}{2}$ optimum items will be selected in expectation. Therefore the total value of the selected optimum items does not exceed $\frac{\epsilon}{2}$ fraction of the optimum solution. 

It suffices to prove that any subset of $k$ items from $\{A_i\}_{i=k+1}^n$ do not have value more than $1 - \frac{1}{e} + \frac{\epsilon}{2}$ fraction of the optimum solution. 
Consider a subset of $k$ items $\{A_{a_i}\}_{i=1}^k$, where $k < a_i \leq n$ for any $1 \leq i	 \leq k$. Each element in the ground set $\UU$ belongs to at least one of these sets with probability $1 - (1-\frac{1}{k})^k \leq 1-\frac{1}{e} + \frac{1}{k} \leq 1- \frac{1}{e} + \frac{\epsilon}{8}$. The probability that there are more than $(1- \frac{1}{e} + \frac{\epsilon}{4})|\UU|$ items in the union of these $k$ sets is not more than $e^{-\frac{(\epsilon/8)^2(1-\frac{1}{e})|\UU|}{3}} \leq \frac{\epsilon}{4} e^{-n\ln(k)}$ using Chernoff bound (Lemma \ref{lem:ChernoffBound}). 
We note that we need $|\UU|$ to be greater than  $\frac{305 n \ln(k) \ln(1/\epsilon)}{\epsilon^2}$.
We know that there are less than ${n \choose k} \leq e^{n\ln(k)}$ size $k$ subsets of $\{A_{k+1}, A_{k+2}, \cdots, A_n\}$. Using union bound, we imply that with probability at least $1 - \frac{\epsilon}{4}$, the union of any $k$ sets in $\{A_{k+1}, A_{k+2}, \cdots, A_n\}$ has size at most $(1-\frac{1}{e} + \frac{\epsilon}{4})|\UU|$. We conclude that the expected value of the best core-set among the selected items is not more than $\frac{\epsilon}{2} |\UU| + (1 - \frac{1}{e} + \frac{\epsilon}{4})|\UU| + \frac{\epsilon}{4}|\UU| = (1 - \frac{1}{e} + \epsilon)|\UU|$. 
\end{proof}

\section{Better Randomized Core-sets for Monotone Submodular Maximization}\label{sec:linear}
In this section, we prove that  although it is not possible to beat the $1-\frac{1}{e}$ barrier, we can slightly increase  the output sizes (to $k'=(\sqrt{2}+1)k$), and apply algorithm $\greedy$ to achieve an approximation factor $\approx 2-\sqrt{2} > \frac{1}{2}$ with a constant multiplicity.
Furthermore, we show in Theorem~\ref{thm:585tight} that our analysis is tight for algorithm $\greedy$ even if we increase the core-set sizes significantly.  
Finally, we present in Subsection~\ref{subsec:PseudoGreedy} a post-processing algorithm $\pseudogreedy$ that achieves an overall distributed approximation factor better than $1/2$. In particular, after the first phase, we show how to find a size $k$ subset of the union of selected items 
with expected value at least $(0.545-o(1))f(\opt)$. 
Since in this section, we are dealing with a monotone submodular function $f$, we can assume {\sc WLOG} that $f(\emptyset) = 0$.

\begin{theorem}\label{thm:two_minus_sqrt_two_core-set}
  For any integer $C \geq 1$, any cardinality constraint $k = o(m)$, 
 algorithm $\greedy$ is a $\left(2-\sqrt{2} - O\left(\frac{1}{k} + \frac{\ln(C)}{C}\right)\right)$-approximate randomized composable core-set of multiplicity $C$ and size $k' = (2\sqrt{2}+1)k$ for any monotone submodular function $f$. By letting $C={1\over \epsilon}$, this leads to a randomized composable core-set of approximation factor $0.5857$. 
\end{theorem}

\begin{proof}
Let $D \defeq 2\sqrt{2}+1$, and $k'=Dk$.
Following our notation from Section~\ref{sec:notation}, let $S_i\defeq \greedy(T_i)$ for $1\le i\le m$, where
$T_i$ is the set of items sent to the machine $i$.
Note that we can let $\vert S_i\vert = Dk$, since if there are less than $Dk$ items in $T_i$, WLOG we can assume the algorithm returns some extra dummy items just for the sake of analysis. 

Consider an item $x\in \opt$. We say that $x$ {\em survives from machine $i$}, if, when we send $x$ to machine $i$ in addition to items of $T_i$, algorithm $\greedy$ would choose this item $x$ in its output of size $k'$, i.e., if $x\in \greedy(T_i \cup \{x\})\}$.
For the sake of analysis, we partition the optimum solution into two sets as follows:
let $\optMS$ be the set of items in the optimum solution that would survive the first machine, i.e., $\optMS\defeq \{x \vert x \in \opt \cap \greedy(T_1 \cup \{x\})\}$. 
Let $\optMNS \defeq \opt \setminus \optMS$, and
$k_1\defeq\vert \optMS\vert$, and $k_2\defeq\vert \optMNS\vert$ (note that $k_1+k_2 = k$).  
We also define $\optMS'\defeq \optMS  \cap \opt^S$, and $\opt' \defeq \optMS' \cup \optMNS$ where $\opt^S$ is defined in Subsection~\ref{sec:notation}. 

We aim to prove that
$\E[f_k(\cup_{i=1}^m S_i)]$ is at least $\left( 2-\sqrt{2} - O\left(\frac{1}{k} + \frac{\ln(C)}{C}\right)\right) f(\opt)$.
Since $f_k$ is a monotone function, we have $f_k(S_1 \cup \optMS') \leq f_k(\cup_{i=1}^m S_i)$. Note that $\optMS'$ is by definition a subset of $\opt^S$, and consequently a subset of $\cup_{i=1}^m S_i$. 
So it suffices to prove that $\E[f_k(S_1 \cup \optMS')]$ is at least  $\left(2-\sqrt{2} - O\left(\frac{1}{k} + \frac{\ln(C)}{C}\right)\right) f(\opt)$. To do so, we first prove the following Lemma:

\begin{lemma}\label{lem:random_machine_selected_vs_total_selected}
The expected value of set $\E[f(\opt')]$ is at least $\left(1-O\left(\frac{\ln(C)}{C}\right)\right) f(\opt)$.
\end{lemma}

\begin{proof} 
Let $\pi$ be an arbitrary fixed permutation on items of $\opt$. For any item $x \in \opt$, we define $\pi^x$ to be all items that appear prior to $x$ in $\pi$.  
 To prove the claim of this lemma, it suffices to show that $f(\opt) - \E[f(\opt')]$ is $O\left(\frac{\ln(C)}{C} f(\opt)\right)$ where the expectation is taken over the distribution of random clustering of items $\{T_i\}_{i=1}^m$. Equivalently, we can prove $\frac{f(\opt) - \E[f(\opt')]}{f(\opt)}$ is $O\left(\frac{\ln(C)}{C}\right)$. 
 At first, we characterize both terms $f(\opt)$, and $f(\opt) - f(\opt')$ in terms of some $\Delta$ values. 
 
 \begin{claim}
 The optimum value $f(\opt)$ is equal to $\sum_{x \in \opt} \Delta(x, \pi^x)$, and the term $f(\opt) - f(\opt')$ is at most $\sum_{x \in \optMS \setminus \opt^S} \Delta(x, \pi^x)$.
 \end{claim}
 
 \begin{proof}
By definition of $\Delta$ values, we have that: $\sum_{x \in \opt} \Delta(x, \pi^x) = f(\opt) - f(\emptyset) = f(\opt)$. 
Similarly, we have that $f(\opt')$ is equal to $\sum_{x \in \opt'} \Delta(x,\pi^x \cap \opt')$. 
By submodularity, we know that $\Delta(x,\pi^x \cap \opt')$ is at least $\Delta(x,\pi^x)$ because $\pi^x \cap \opt'$ is a subset of $\pi^x$. Therefore, we have: 

\begin{eqnarray*}
f(\opt) - f(\opt') 
&=& 
\sum_{x \in \opt} \Delta(x,\pi^x) - \sum_{x \in \opt'} \Delta(x,\pi^x \cap \opt') \\
&\leq & 
\sum_{x \in \opt} \Delta(x,\pi^x) - \sum_{x \in \opt'} \Delta(x,\pi^x) \\
&=& %
\sum_{x \in \opt \setminus \opt'} \Delta(x,\pi^x)
= \sum_{x \in \optMS \setminus \opt^S} \Delta(x,\pi^x)
\end{eqnarray*}

where the last equality holds by definition of $\opt'$. 
\end{proof}

Now, it suffices to upper bound the expected value of $\E[\sum_{x \in \optMS \setminus \opt^S} \Delta(x,\pi^x)]$ by $O\left(\frac{\ln(C)}{C}\right) f(\opt)$ as follows: 

$$
\E[\sum_{x \in \optMS \setminus \opt^S} \Delta(x,\pi^x)] = \sum_{x \in \opt} Pr[x \in \optMS \setminus \opt^S] \Delta(x,\pi^x)
$$

We note that $\Delta(x,\pi^x)$ is a fixed (non-random) term and therefore we can take it out of the expectation. Since $f(\opt)$ is equal to $\sum_{x \in \opt} \Delta(x,\pi^x)$, we just need to prove that $Pr[x \in \optMS \setminus \opt^S]$ is at most $O\left(\frac{\ln(C)}{C}\right)$ for any item $x \in \opt$. 

We note that the first machine is just a random machine, and the distribution of set of items sent to it, $T_1$, is the same as any other set $T_i$ for any $2 \leq i \leq m$.  
We consider two cases for an item $x \in \opt$: 
\begin{itemize}
\item
The probability of $x$ being chosen when added to a random machine is at most $\frac{\ln(C)}{C}$, i.e. $Pr[x \in \greedy(T_1 \cup \{x\})] = \frac{\ln(C)}{C}$. 
\item
The probability $Pr[x \in \greedy(T_1 \cup \{x\})]$ is at least $\frac{\ln(C)}{C}$. 
\end{itemize}

In the first case, we know that $Pr[x \in \optMS] \leq \frac{\ln(C)}{C}$, and therefore   $Pr[x \in \optMS \setminus \opt^S] \leq \frac{\ln(C)}{C}$ which concludes the proof. 

In the latter case, we prove that $Pr[x \notin \opt^S] = O(\frac{\ln(C)}{C})$ which implies the claim of the lemma as $Pr[x \notin \opt^S] \geq Pr[x \in \optMS \setminus \opt^S]$.
Let $i_1, i_2, \cdots, i_C$ be the indices of the $C$ random machines where $x$ is sent to, i.e. $x \in T_{i_{\ell}}$ for $1 \leq \ell \leq C$. 
Item $x$ is not in $\opt^S$ if it is not selected in non of these $C$ machines.
Formally, $Pr[x \notin \opt^S]$ is equal to $Pr[x \notin \greedy(T_{i_{\ell}}),\ \forall 1 \leq \ell \leq C]$.

We note that for any $1 \leq i \leq m$, and each item $y \in \NN \setminus \{x\}$, 
the probability $Pr[y \in T_i]$ is $\frac{C}{m}$. 
Although any pair of sets $T_i$ and $T_{i'}$ are correlated (where $1 \leq i, i' \leq m$), the 
events $x' \in T_i$ and $x'' \in T_i$ are independent for any pair of distinct items $x', x'' \in \NN \setminus \{x\}$.   
So for any $1 \leq \ell \leq C$, the distribution of $T_{i_{\ell}} \setminus \{x\}$ is the same as $T_1 \setminus \{x\}$. Therefore the probability that machine $i_{\ell}$ selects item $x$, $Pr[x \in \greedy(T_{i_{\ell}})]$, is 
equal to $Pr[x \in \greedy(T_1 \cup \{x\})]$. Since we are considering the latter case, we imply that $Pr[x \in \greedy(T_{i_{\ell}})]$ is at least $\frac{\ln(C)}{C}$ for any $1 \leq \ell \leq C$. 
If the $C$ sets $\{T_{i_{\ell}} \setminus \{x\}\}_{\ell=1}^C$ were not correlated (mutually independent), we could say that events $x \in \greedy(T_{i_{\ell}})$ are independent. Consequently, we would have $Pr[x \notin \greedy(T_{i_{\ell}}),\ \forall 1 \leq \ell \leq C]$ is at most $\left(1-\frac{\ln(C)}{C}\right)^C \leq e^{-\ln(C)} = \frac{1}{C}$, and therefore $Pr[x \in \optMS \setminus \opt^S] \leq \frac{1}{C}$ which completes the proof. 
Although these sets are correlated, we show that this correlation is bounded as follows.  

We start by constructing $C$ sets $\{T'_{i_{\ell}}\}_{\ell=1}^C$ which are very similar to the sets $\{T_{i_{\ell}}\}_{\ell=1}^C$ with the exception that the sets $\{T'_{i_{\ell}} \setminus \{x\}\}_{\ell=1}^C$ are mutually independent. 
We exploit this independence property to show that
$Pr[x \notin \greedy(T'_{i_{\ell}}),\ \forall 1 \leq \ell \leq C] \leq \frac{1}{C}$. 
Finally, we use a coupling technique to show that with high probability ($1-O(\frac{1}{C})$), the solution of algorithm $\greedy$ on $T_{i_{\ell}}$ is the same as its solution on $T'_{i_{\ell}}$ for any $1 \leq \ell \leq C$, i.e. $\greedy(T_{i_{\ell}}) = \greedy(T'_{i_{\ell}})$ which completes the proof. Following we elaborate on the construction of sets $\{T'_{i_{\ell}}\}_{\ell=1}^C$, and how they are coupled with sets $\{T_{i_{\ell}}\}_{\ell=1}^C$. 
 
For any $1 \leq \ell \leq C$, 
insert each item $y \in \NN \setminus \{x\}$ in $T'_{i_{\ell}}$ with probability $\frac{C}{m}$ independently. Insert $x$ into all these $C$ sets
$\{T'_{i_{\ell}}\}_{\ell=1}^C$.  
We note that for each item $y \in \NN \setminus \{x\}$, the $C$ events $y \in T'_{i_{\ell}}$ are independent, and consequently sets $\{T'_{i_{\ell}} \setminus \{x\}\}_{\ell=1}^C$ are mutually independent.
Since for each $1 \leq \ell \leq C$, set $T'_{i_{\ell}} \setminus \{x\}$ has the same distribution as $T_1 \setminus \{x\}$, and also using the independence property, 
we have that $Pr[x \notin \cup_{\ell=1}^C \greedy(T'_{i_{\ell}})] \leq (1-\frac{\ln(C)}{C})^C \leq \frac{1}{C}$. We are ready to explain the coupling technique. 

We show how to construct sets $\{T_{i_{\ell}}\}_{\ell=1}^C$ from sets $\{T'_{i_{\ell}}\}_{\ell=1}^C$.
We start with $\ell=1$, and increase it one by one. 
The first step is easy. 
We set $T_{i_1}$ to be equal to $T'_{i_1}$.
For any $1 < \ell \leq C$,  
we initialize $T_{i_{\ell}}$ by $T'_{i_{\ell}}$, and we adjust $T_{i_{\ell}}$ by a few item insertions and deletions to address the correlations between $T_{i_{\ell}}$ and prior sets $T_{i_1}, T_{i_2}, \cdots, T_{i_{\ell-1}}$. 
For any $1 \leq \ell \leq C$, 
we also upper bound the size of
$T_{i_{\ell}} \Delta T'_{i_{\ell}}$ with high probability. 
For any $y \in \NN \setminus \{x\}$, if $y$ has appeared $a$ times in $T_{i_1}, T_{i_2}, \cdots, T_{i_{\ell-1}}$, it should be present in $T_{i_{\ell}}$ with probability $\frac{C-a}{m-\ell+1}$ instead of $\frac{C}{m}$. 
This probability is at most $\frac{C}{m-C+1}$ which is less than $\frac{C}{m} + \frac{C^2}{m^2}$ for $m \geq C^2$. 
So the presence probability of each item is not increased by more than $\frac{C^2}{m^2}$ from the initial probability $\frac{C}{m} = Pr[y \in T'_{i_{\ell}}]$. To increase the presence probability of $y$ by some $\delta$, it suffices to insert it to $T_{i_{\ell}}$ with probability $\frac{\delta}{1-C/m}$ because after this extra probabilistic insertion, it will be in $T_{i_{\ell}}$ with probability $\frac{C}{m} + (1-\frac{C}{m})\frac{\delta}{1-C/m} = \frac{C}{m} + \delta$.  Therefore, we have that $Pr[y \in T_{i_{\ell}} \setminus T'_{i_{\ell}}]$ is at most $\frac{C^2}{m^2(1-C/m)} \leq \frac{2C^2}{m^2}$. 
 
On the other hand, the probability of $y$ being present in $T_{i_{\ell}}$ should decrease only if $y$ appears in $T_{i_{\ell'}}$ for some $1 \leq \ell' < \ell$. 
To reduce the probability of $y$ being present in $T_{i_{\ell}}$ by some $\delta$, it suffices to delete it from $T_{i_{\ell}}$ (when it is in $T_{i_{\ell}}$) with probability $\frac{\delta}{C/m}$. 
Therefore, $y$ is in $T'_{i_{\ell}} \setminus T_{i_{\ell}}$ only if $y$ is in intersection $T_{i_{\ell'}} \cap T'_{i_{\ell}}$ for some $\ell' < \ell$. 
The probability $Pr[y \in T_{i_{\ell'}} \cap T'_{i_{\ell}}]$ is equal to 
$Pr[y \in T_{i_{\ell'}}]Pr[y \cap T'_{i_{\ell}}] = \frac{C}{m} \times \frac{C}{m} =\frac{C^2}{m^2}$ (note that sets $T_{i_{\ell'}}$, and $T'_{i_{\ell}}$ are independent). 
Since there are at most $C$ choices for $\ell'$, the probability $Pr[y \in T'_{i_{\ell}} \setminus T_{i_{\ell}}]$ is upper bounded by $\frac{C^3}{m^2}$.

By this coupling technique, we conclude that for any item $y \in \NN \setminus \{x\}$, the probability $Pr[y \in T_{i_{\ell}} \Delta T'_{i_{\ell}}]$ is at most $\frac{3C^3}{m^2}$. Since there are $n$ items in total, and these events are independent for different items, we can use Chernoff bound \ref{lem:ChernoffBound}to have:

$$
Pr[|T_{i_{\ell}} \Delta T'_{i_{\ell}}| > 	4+(C+1)\frac{3C^3}{m^2}n] \leq  e^{4(1-\ln(C))} \leq C^{-2}
$$

We also know that $T_{i_{\ell}}$ has expected size $\frac{n}{m} \geq k$, and by Chernoff bound \ref{lem:ChernoffBound} with probability at least $1-e^{-k/8} = 1-O(C^{-2})$, its size is not less than $\frac{n}{2m}$ (note that $C \leq k$). For every $\ell$, we conclude that with probability $1-O(\frac{1}{C^2})$, we have: 

\begin{eqnarray*}
|T_{i_{\ell}}| &\geq& \frac{n}{2m} \\
|T_{i_{\ell}} \Delta T'_{i_{\ell}}| &\leq& 	4+(C+1)\frac{3C^3}{m^2}n \leq 4+ \frac{6C^4}{m} \times \frac{n}{m}
\end{eqnarray*}

We are ready to prove that with probability $Pr[\greedy(T_{i_{\ell}}) = \greedy(T'_{i_{\ell}}),\ \forall 1 \leq \ell \leq C]$ is $1-O(\frac{1}{C})$. 
It suffices to show for every $\ell$ that $Pr[\greedy(T_{i_{\ell}}) = \greedy(T'_{i_{\ell}})]$ is $1-O(\frac{1}{C^2})$ using the Union bound for these $C$ events.
Since $m$ is at least $12C^6Dk$ for constant value of $C$, with probability $1-O(\frac{1}{C^2})$, the ratio $\frac{|T_{i_{\ell}} \Delta T'_{i_{\ell}}|}{|T_{i_{\ell}}|}$ is at most $\frac{1}{C^2Dk}$.
We note that $T_{i_{\ell}} \setminus T'_{i_{\ell}}$ is a random subset of $T_{i_{\ell}}$, and its selection has nothing to do with $f$ values. Therefore the probability that $T_{i_{\ell}} \setminus T'_{i_{\ell}}$ has some non-empty intersection with set $\greedy(T_{i_{\ell}})$ is equal to $\frac{\lvert T_{i_{\ell}} \setminus T'_{i_{\ell}} \rvert \times \lvert \greedy(T_{i_{\ell}})\rvert}{\lvert T_{i_{\ell}} \rvert} \leq \frac{1}{C^2}$.
 Assuming $T_{i_{\ell}} \setminus T'_{i_{\ell}}$ has no common item with 
 $\greedy(T_{i_{\ell}})$, we can say that $\greedy(T_{i_{\ell}} \cap T'_{i_{\ell}})$ is the same as $\greedy(T_{i_{\ell}})$ because removing some items from $T_{i_{\ell}}$ that algorithm $\greedy$ does not select will not change the output of algorithm $\greedy$ on it (it is the first property of $\beta$-nice algorithms explained in Definition~\ref{def:nice}). Formally, with probability $1-O(\frac{1}{C^2})$, the intersection $\left(T_{i_{\ell}} \setminus T'_{i_{\ell}}\right) \cap \greedy(T_{i_{\ell}})$ is empty, and consequently, we have $\greedy(T_{i_{\ell}}) = \greedy\left(T_{i_{\ell}} \setminus \left(T_{i_{\ell}} \setminus T'_{i_{\ell}} \right)\right) = \greedy(T_{i_{\ell}} \cap T'_{i_{\ell}})$. Similarly, we can prove that with probability $1-O(\frac{1}{C^2})$, sets $\greedy(T'_{i_{\ell}})$, and $\greedy(T_{i_{\ell}} \cap T'_{i_{\ell}})$ are the same. We conclude that with probability $1-O(\frac{1}{C^2})$, the sets $\greedy(T_{i_{\ell}})$, and $\greedy(T'_{i_{\ell}})$ are the same. Using union bound, we have this equality for all $\ell$ with probability $1-O(\frac{1}{C})$. We already know that $Pr[x \in \cup_{\ell=1}^C T'_{i_{\ell}}]$ is at least $1 - \frac{1}{C}$. Therefore $Pr[x \in \opt^S] = Pr[x \in \cup_{\ell=1}^C T_{i_{\ell}}]$ is $1-O(\frac{1}{C})$ which concludes the proof of lemma. 
\end{proof}

Using Lemma~\ref{lem:random_machine_selected_vs_total_selected}, it is sufficient prove that ratio $\frac{f_k(S_1 \cup \optMS')}{f(\opt')}$ is at least $2-\sqrt{2} - O\left(\frac{1}{k}\right)$. 
In order to lower bound the ratio $\frac{f_k(S_1 \cup \optMS')}{f(\opt')}$, we write the following factor-revealing linear program  $LP^{k, k_2}$, and prove in Lemma~\ref{lem:LP_lower_bound_core-set} that the solution to this LP is a lower bound on the aforementioned ratio.

\begin{center}
\begin{tabular}{|ll|} 
\hline 
$LP^{k,k_2}$ Minimize $\beta$ Subject to: & \\ \hline 
$(1)$ $\beta \geq 1-\alpha +\sum_{j \in J} a_j + c_j$  $\forall J \subset [Dk] \& |J| = k_2$ &
$(4)$ $0 \leq \alpha, a_j, b_j, c_j \leq 1$ $\ \ \ \ \ \ \ \ \ \ \ \ \ \ \ \ \ \ \ \ \ \ \ $ $\forall 1 \leq j \leq Dk$ \\ 
$(2)$ $a_j + b_j + c_j \geq \frac{\alpha - \sum_{j'=1}^{j-1}a_{j'}}{k_2}$  $\forall 1 \leq j \leq Dk$ &
$(5)$ $\beta \geq \sum_{j=1}^{k}  a_j + b_j + c_j$  \\
$(3)$ $\sum_{j=1}^{Dk}  b_j \leq 1- \alpha$  &
$(6)$ $a_j + b_j + c_j \geq a_{j+1}+b_{j+1}+c_{j+1}$  $\forall 1 \leq j < Dk$\\ \hline
\end{tabular} 
\end{center}

\begin{lemma}\label{lem:LP_lower_bound_core-set}
For any integer $k  > 0$, the ratio $\frac{f_k(S_1 \cup \optMS')}{f(\opt')}$ is lower bounded by the optimum solution of linear program $LP^{k,k_2}$ for some integer $1 \leq k_2 \leq k$. 
\end{lemma}

\begin{proof}
We want to prove that ratio $\frac{f_k(S_1 \cup \optMS')}{f(\opt')}$ is lower bounded by the solution of minimization linear program $LP^{k,k_2}$. It suffices to construct one feasible solution with objective value $\beta$ equal to $\frac{f_k(S_1 \cup \optMS')}{f(\opt')}$. 
At first, we construct this solution  for every instance of the problem, and then prove its feasibility in Claim~\ref{claim:LP_feasibility}.  

We remind that $\opt'$ is the union of two disjoint sets $\optMS'$, and $\optMNS$. Fix a permutation $\pi$ on the items of $\opt'$
such that every item of $\optMS'$ appears before every item of $\optMNS$ in $\pi$. In other words, $\pi$ is an arbitrary permutation on items of $\optMS'$ followed by an arbitrary permutation on items of $\optMNS$. 
For any item $x$ in $\opt'$, define $\pi^x$ to be the set of items in $\opt'$
that appear prior to $x$ in permutation $\pi$.	
For any $1 \leq j \leq Dk$, we define set $S^j$ to be the first $j$ items of $S_1$.
We set the linear program variables as follows: 

\begin{eqnarray*}
\beta  &\defeq&  \frac{f_k(S_1 \cup \optMS')}{f(\opt')} \\ 
\alpha &\defeq& \frac{\sum_{x \in \optMNS} \Delta(x,\pi^x)}{f(\opt')}\\
a_j      &\defeq& \frac{\sum_{x \in \optMNS} \Delta(x,\pi^x \cup S^{j-1}) - \Delta(x,\pi^x \cup S^{j})}{f(\opt')} \\
b_j      &\defeq& \frac{\sum_{x \in \optMS'} \Delta(x,\pi^x \cup S^{j-1}) - \Delta(x,\pi^x \cup S^{j})}{f(\opt')} \\
c_j      &\defeq& \frac{f(S^{j}) - f(S^{j-1})}{f(\opt')} - a_j - b_j \\
\end{eqnarray*}

\begin{claim}\label{claim:LP_feasibility}
The above assignment forms a feasible solution of $LP^{k,k_2}$, and its solution is equal to $\frac{f_k(S_1 \cup \optMS')}{f(\opt')}$. 
\end{claim}

\begin{proof}
The claim on the objective value of the solution is evident by definition of $\beta$. 
To prove that the constraints hold, we show some simple and useful facts about the marginal values of items in $\opt$. We note that $\sum_{x \in \opt'} \Delta(x,\pi^x) = f(\opt')$ by definition of $\Delta$ values and the fact that $f(\emptyset)=0$. 
Similarly for any $1 \leq j \leq Dk$, we have that: 

\begin{eqnarray}
\label{eq:sum_marginals}
\sum_{x \in \opt'} \Delta(x, \pi^x \cup S^j) = f(\opt' \cup S^j) - f(S^j)
\end{eqnarray}

We are ready to prove that all constraints $(1), (2), \cdots, (5)$ one by one. We start with constraint $(1)$.
For any set $J \subset [DK]$ with size $|J|=k_2$, we  
define $S(J)$ to be $\{y_j | j \in J\}$ where $y_j$ is the $j$th item selected by algorithm $\greedy$ in $S_1$. We also define $S'(J)$ to be $\optMS' \cup S(J)$. 
Set $S'(J)$ is  a subset of $\optMS' \cup S_1$ with size at most $k_2+k_1=k$. Therefore $f(S'(J))$ is a lower bound on $f_k(\optMS' \cup S_1)$. We can also lower bound $f(S'(J))$ as follows: 

\begin{eqnarray*}
&f(S'(J))& = \sum_{x \in \optMS'} \Delta(x, \pi^x \cap \optMS') + \sum_{j \in J} \Delta(y_j, \optMS' \cup (S^{j-1} \cap S(J))) \\
&\geq& 
\sum_{x \in \optMS'} \Delta(x, \pi^x) + \sum_{j \in J} \Delta(y_j, \opt' \cup S^{j-1}) \\
&=&
(1-\alpha)f(\opt') + \sum_{j \in J} \left( f(\optMS' \cup S^j) - f(\optMS' \cup S^{j-1})\right) \\
&=&
(1-\alpha)f(\opt') + \\
&&\sum_{j \in J} \bigg( f(S^j) - f(S^{j-1}) + \Big(f(\optMS' \cup S^j) - f(S^j)\Big) - \Big(f(\optMS' \cup S^{j-1}) - f(S^{j-1}) \Big)\bigg) \\
&=&
(1-\alpha)f(\opt') + \\
&&\sum_{j \in J} \bigg( (a_j+b_j+c_j)f(\opt') + \Big(f(\optMS' \cup S^j) - f(S^j)\Big) - \Big(f(\optMS' \cup S^{j-1}) - f(S^{j-1}) \Big)\bigg) 
\end{eqnarray*}

The first equality holds by definition of $\Delta$. The first inequality holds by submodularity of $f$, and knowing that $S^{j-1} \cap S(J) \subseteq S^{j-1}$. The second equality holds by definition of $\alpha$, and the last equality holds by definition of $c_j$. We claim that $\Big(f(\optMS' \cup S^j) - f(S^j)\Big) - \Big(f(\optMS' \cup S^{j-1}) - f(S^{j-1}) \Big)$ (which is part of the right hand side of the last equality) is equal to $-b_j f(\opt')$. We note that $\Big(f(\optMS' \cup S^j) - f(S^j)\Big)$ is equal to $\sum_{x \in \optMS'} \Delta(x, \pi^x \cup S^j)$, and similarly $\Big(f(\optMS' \cup S^{j-1}) - f(S^{j-1})\Big)$ is equal to $\sum_{x \in \optMS'} \Delta(x, \pi^x \cup S^{j-1})$. By taking the difference of them, we have: 

\begin{eqnarray*}
\Big(f(\optMS' \cup S^j) - f(S^j)\Big) - \Big(f(\optMS' \cup S^{j-1}) - f(S^{j-1}) \Big) &=& \sum_{x \in \optMS'} \Delta(x, \pi^x \cup S^j) - \Delta(x, \pi^x \cup S^{j-1})
\end{eqnarray*}

which is (by definition of $b_j$) equal to $-b_j f(\opt')$. We conclude that $f(S')$, and consequently $f_k(S_1 \cup \optMS')$ are both at least $(1-\alpha + \sum_{j \in J} a_j + c_j) f(\opt')$ which concludes the proof of constraint $(1)$.

We prove  constraint $(2)$ using the fact that algorithm $\greedy$ selects the item with maximum marginal value in each step. We note that the right hand side of constraint $(2)$ is $a_j+b_j+c_j$ which is by definition the marginal gain of item $y_j$ 
divided by $f(\opt')$, i.e. $\frac{\Delta(y_j, S^{j-1})}{f(\opt')}$. 
We know that any item $x \in \optMNS$ will not be selected by algorithm $\greedy$ if it is part of the input set which means that the marginal gain $(a_j+b_j+c_j)f(\opt') = \Delta(y_j, S^{j-1})$ is at least the marginal gain $\Delta(x,S^{j-1})$ for any $x \in \optMNS$, and it is also greater than the average of these marginal gains. In other words, we have:

$$
a_j+b_j+c_j \geq \frac{\sum_{x \in \optMNS} \Delta(x,S^{j-1})}{\lvert \optMNS\rvert f(\opt')} = 
\frac{\sum_{x \in \optMNS} \Delta(x,S^{j-1})}{k_2 f(\opt')}
$$  

To finish the proof of constraint $(2)$, it suffices to prove the following inequality:

\begin{eqnarray}\label{eq:constraint2}
\sum_{x \in \optMNS} \Delta(x,S^{j-1}) \geq
f(\opt') \left(\alpha - \sum_{j'=1}^{j-1}a_{j'} \right)
\end{eqnarray}

By definition of $\alpha$, and 
$a$ values, we have that: 

\begin{eqnarray*}
f(\opt') \left( \alpha - \sum_{j'=1}^{j-1}a_{j'} \right)
&=& \sum_{x \in \optMNS} \Delta(x,\pi^x) - \sum_{j'=1}^{j-1} \sum_{x \in \optMNS} \left(\Delta(x,\pi^x \cup S^{j'-1}) - \Delta(x, \pi^x \cup S^{j'}) \right) \\
&=&
\sum_{x \in \optMNS} \Delta(x,\pi^x) - \sum_{x \in \optMNS} \sum_{j'=1}^{j-1} \left( \Delta(x,\pi^x \cup S^{j'-1}) - \Delta(x, \pi^x \cup S^{j'}) \right) \\
&=&
\sum_{x \in \optMNS} \Delta(x,\pi^x) - \sum_{x \in \optMNS} \left( \Delta(x,\pi^x) - \Delta(x, \pi^x \cup S^{j-1}) \right) \\
&=& \sum_{x \in \optMNS} \Delta(x,\pi^x \cup S^{j-1})
\end{eqnarray*}
  
which completes the proof of Equation~\ref{eq:constraint2}, and consequently constraint $(2)$.

Now, we prove that constraint $(3)$ holds. By definition of $\alpha$, and the fact that $f(\opt') =$ $\sum_{x \in \opt'}$ $\Delta(x, \pi^x)$, we know that $1-\alpha$ is equal to $\frac{\sum_{x \in \optMS'} \Delta(x,\pi^x)}{f(\opt')}$. We also know that:
 
\begin{eqnarray*}
\sum_{j=1}^{Dk} b_j &=& \frac{\sum_{j=1}^{Dk} \sum_{x \in \optMS'} \Delta(x, \pi^x \cup S^{j-1}) - \Delta(x, \pi^x \cup S^j)}{f(\opt')} \\
&=& \frac{\sum_{x \in \optMS'} \Delta(x, \pi^x) - \Delta(x, \pi^x \cup S^{Dk})}{f(\opt')} 
\leq  \frac{\sum_{x \in \optMS'} \Delta(x, \pi^x)}{f(\opt')} = 1 -\alpha
 \end{eqnarray*}

where the inequality holds because valuation function $f$ is monotone, and therefore all $\Delta$ values are non-negative.  This proves that constraint $(3)$ holds.

To prove constraint $(4)$, we should show that variables $a_j, b_j$, $c_j$, and $\alpha$ are all in range $[0,1]$ for any $1 \leq j \leq DK$.
We know that $f(\opt')$ is equal to $\sum_{x \in \opt'} \Delta(x, \pi^x)$. Therefore by definition, $\alpha$ is equal to $\frac{\sum_{x \in \optMNS} \Delta(x, \pi^x)}{\sum_{x \in \opt'} \Delta(x, \pi^x)}$
. Since  $\optMNS$ is a subset of $\opt'$, we imply that $\alpha$ is at most $1$. We also know that $\Delta$ values are all non-negative, and therefore $\alpha$ is non-negative.
Now, we should prove that variables $a_j, b_j$, and $c_j$ are all in range $[0,1]$.
By definition, $a_j+b_j+c_j$ is equal to $\frac{\Delta(y_j,S^{j-1})}{f(\opt')} \leq \frac{f(\{y_j\}) - f(\emptyset)}{f(\opt')} = \frac{f(\{y_j\})}{f(\opt')}$ where the inequality and equality are implied by the submodularity of $f$, and the fact $f(\emptyset) = 0$ respectively. If $f(\{y_j\})$ is at least $f(\opt')$, the proof of Lemma~\ref{lem:LP_lower_bound_core-set} can be completed as follows. 
We know that $f_k(S_1 \cup \optMS') \geq f(\{y_j\})$, and therefore the ratio $\frac{f_k(S_1 \cup \optMS')}{f(\opt')}$ is at least $1$. On the other hand, there exists a very simple solution for $LP^{k_2, k}$ with objective value $\beta=1$ by just setting all variables to zero, and $\beta$ equal to one which completes the proof in the case $f(\{y_j\}) \geq f(\opt')$. So we can focus on the case, $f(\{y_j\}) \leq f(\opt')$ in which we have $a_j+b_j+c_j \leq \frac{f(\{y_j\})}{f(\opt')} \leq 1$. 
So it suffices to prove that these three variables are non-negative to prove that constraint $(4)$  holds. Variables $a_j$ and $b_j$ are non-negative because $f$ is submodular, and $S^{j-1}$ is a subset of $S^j$. We use Equation~\ref{eq:sum_marginals} to prove non-negativity of $c_j$. By definition, $a_j+b_j$ is equal to $\frac{\sum_{x \in \opt'} \Delta(x, \pi^x \cup S^{j-1}) - \Delta(x, \pi^x \cup S^j)}{f(\opt')}$. By applying Equation~\ref{eq:sum_marginals}, we have: 

$$
a_j + b_j = \frac{f(\opt' \cup S^{j-1}) - f(S^{j-1}) - f(\opt' \cup S^j) + f(S^j)}{f(\opt')}
$$

This implies that:

\begin{eqnarray*}
c_j &=& \frac{f(S^{j}) - f(S^{j-1}) -f(\opt' \cup S^{j-1}) + f(S^{j-1}) + f(\opt' \cup S^j) - f(S^j)}{f(\opt')} \\
&\geq& 
\frac{ -f(\opt' \cup S^{j-1})  + f(\opt' \cup S^j)}{f(\opt')} \geq 0
\end{eqnarray*}
 
 where the last inequality holds because of monotonicity of $f$.

  We prove constraint $(5)$ as follows. At first, we show that the right hand side of constraint $(5)$ is simply equal to $\frac{f(S^{k})}{f(\opt')}$. We know that $a_j+b_j+c_j = \frac{f(S^j) - f(S^{j-1})}{f(\opt')}$ for each $1 \leq j \leq k$. By a telescopic summation, we have that the right hand side of constraint $(5)$, $\sum_{j=1}^k a_j+b_j+c_j$, is equal to $\frac{f(S^k) - f(\emptyset)}{f(\opt')} = \frac{f(S^k)}{f(\opt')}$. By definition of $\beta$, and the fact that $f_k(S_1 \cup \optMS')$ is at least $f(S^k)$, we conclude that constraint $(5)$ holds.

  To prove constraint $(6)$, we note that by definition, $a_j + b_j +c_j$ is $\Delta(y_j, S^{j-1})$. 
Since algorithm $\greedy$ chooses the item with maximum marginal value at each step, we have $\Delta(y_j, S^{j-1}) \geq \Delta(y_{j+1}, S^{j-1})$. By submodularity, we have $\Delta(y_{j+1}, S^{j-1}) \geq \Delta(y_{j+1}, S^j) = a_{j+1} + b_{j+1} + c_{j+1}$. We conclude that constraint $(6)$ holds. 
  Therefore the proofs of Claim~\ref{claim:LP_feasibility}, and Lemma~\ref{lem:LP_lower_bound_core-set} are also complete.
\end{proof}
\end{proof}

Finally, we show that the solution of $LP^{k,k_2}$ is at least $2-\sqrt{2} - O(\frac{1}{k})$ for any possible value of $k_2$ which concludes the proof of Theorem~\ref{thm:two_minus_sqrt_two_core-set}.  
 
\begin{lemma}\label{lem:lower_bound_LP}
The optimum solution of linear program $LP^{k,k_2}$ is at least  $2 - \sqrt{2} - O(\frac{1}{{k}})$ for any $0 \leq k_2 \leq k$.
\end{lemma}

\begin{proof}
We consider two cases: a) $k_2 \leq \frac{k}{10}$, and b) $k_2 > \frac{k}{10}$. We first consider the former case which is easier to prove, and then focus on the latter case. 
If $k_2$ is at most  $\frac{k}{10}$, we prove that 
 the objective function of linear program $LP^{k,k_2}$ (which is $\beta$) cannot be less than $0.6$ which concludes the proof of this lemma. 
 Since all variables are non-negative, we can apply constraint $(2)$ for each $j$ in range $[1,k]$, and imply that $\sum_{j=1}^k a_j + b_j + c_j$ is at least $\left(1- \left(1 - \frac{1}{k_2} \right)^k\right) \alpha$. Using constraint $(5)$, we know that $\sum_{j=1}^k a_j + b_j + c_j$ is a lower bound for $\beta$. We are also considering the case $k_2 \leq \frac{k}{10}$, therefore $\beta$ is at least $\left(1- \left(1 - \frac{1}{k_2} \right)^k\right) \alpha \geq \left(1-e^{-10}\right)\alpha \geq 0.9999\alpha$. If $\alpha$ is at least $0.586$, the claim of Lemma~\ref{lem:lower_bound_LP} is proved. So we assume $\alpha$ is at most $0.586$. 
 We define three sets of indices: $J_1 \defeq \{1, 2, \cdots, k_2\}$, $J_2 \defeq \{k_2+1, k_2+2, \cdots, 2k_2\}$, and $J_3 \defeq \{2k_2+1, 2k_2+2, \cdots, 3k_2\}$. We note that these sets have size $k_2$, and therefore constraint $(1)$ should hold for them. If there exists some set $J \subset Dk$ with size $k_2$ such that $\sum_{j \in J} a_j + c_j$ is at least $0.3\alpha$, we can use constraint $(1)$ to lower bound $\beta$ with $1-\alpha + 0.3\alpha = 1-0.7\alpha > 2 - \sqrt{2}$. Therefore we can assume that $\sum_{j \in J_i} a_j+c_j$ is at most $0.3\alpha$ for $i \in \{1, 2, 3\}$. Using constraint $(2)$, we imply that:
 
\begin{eqnarray*}
\sum_{j \in J_1} a_j+b_j+c_j &\geq& 0.7\alpha \\
\sum_{j \in J_2} a_j+b_j+c_j &\geq& 0.4\alpha \\
\sum_{j \in J_3} a_j+b_j+c_j &\geq& 0.1\alpha \\ 
\end{eqnarray*}
 
 We also know that $J_1 \cup J_2 \cup J_3 \subset \{1, 2, \cdots, k\}$. We apply constraint $(5)$ to imply that $\beta$ is at least $(0.7+0.4+0.1)\alpha = 1.2\alpha$. This yields a stronger upper bound on $\alpha$. If $\alpha$ is at least $\frac{2-\sqrt{2}}{1.2} < 0.49$, the claim of Lemma~\ref{lem:lower_bound_LP} is proved. So we assume $\alpha \leq 0.49$. We follow the above argument one more time, and the proof is complete. If for some $i \in \{1, 2, 3\}$, the sum $\sum_{j \in J_i} a_j+c_j$ is at least $0.16\alpha$, using constraint $(1)$, we can lower bound $\beta$ with $1-0.84\alpha > 2-\sqrt{2}$. Therefore we have $\sum_{j \in J_i} a_j+c_j \leq 0.16\alpha$ for each $i \in \{1, 2, 3\}$ which yields the following stronger inequalities: 
 
\begin{eqnarray*}
\sum_{j \in J_1} a_j+b_j+c_j &\geq& 0.84\alpha \\
\sum_{j \in J_2} a_j+b_j+c_j &\geq& 0.68\alpha \\
\sum_{j \in J_3} a_j+b_j+c_j &\geq& 0.52\alpha \\ 
\end{eqnarray*}
 
By applying constraint $(5)$, we have $\beta \geq (0.84+0.68+0.52)\alpha > 2\alpha$. We can also use constraint $(1)$, and conclude that $\beta \geq \max\{2\alpha, 1-\alpha\} > 2-\sqrt{2}$ which completes the proof for the former case $k_2 \leq \frac{k}{10}$.

In the rest of the proof, we consider the latter case $k_2 > \frac{k}{10}$.
The structure of the proof is as follows: In Claim~\ref{claim:LP_optimum_solution_structure}, we first show that without loss of generality, one can  assume a special structure in an optimum solution of $LP^{k,k_2}$, and then exploit this structure to show that any solution of $LP^{k,k_2}$ is lower bounded by a simple system of two equations  with some $O(\frac{1}{k_2}) = O(\frac{1}{k})$ error.  
We can explicitly analyze this system of equations, and achieve a lower bound of $2-\sqrt{2}$ on $\beta$ (the key variable of the system of equations, and also the objective function of linear program $LP^{k,k_2}$).

\begin{claim}\label{claim:LP_optimum_solution_structure}
There exists an optimum solution for linear program $LP^{k,k_2}$ with the following three properties: 
\begin{itemize}
\item $c_j = 0$ for all $1 \leq j \leq Dk$
\item $a_j \geq a_{j+1}$ for all $1 \leq j < Dk$
\item constraint $(2)$ is tight for all $1 \leq j \leq Dk$
 \end{itemize}
\end{claim}

\begin{proof}
It suffices to show that every optimum solution of $LP^{k,k_2}$ without changing the objective $\beta$ can be transformed to a feasible solution  with the above three properties. 
Consider an optimum solution $\left(\beta^*, \alpha^*, \{a^*_j, b^*_j, c^*_j\}_{j=1}^{Dk}\right)$. 
We start by showing how $c^*_j$s can be set to zero. 
Suppose $c^*_j > 0$ for some $1 \leq j \leq Dk$. 
We can increase the value of $a^*_j$ by $c^*_j$, and then set $c^*_j$ to zero. This update keeps $a^*_j+c^*_j$ intact, and therefore does not change anything in constraints $(1)$, $(3)$, $(5)$, and $(6)$. It also makes it easier to satisfy constraint $(2)$ since it (possibly) reduces the right hand side, and keeps the left hand side intact. Constraint $(4)$ remains satisfied since $a^*_j + c^*_j \leq 1$ (otherwise $\beta$ is also at least $1$ which proves the claim of Lemma~\ref{lem:lower_bound_LP} directly). Therefore we can assume $c$ variables are equal to zero, and exclude them to have a simpler linear program: 

\vspace{0.2in}
\begin{tabular}{|ll|} \hline $LP^{k,k_2}$ Minimize $\beta$ & \\ \hline 
Subject to: & \\
$(1)$ $\beta \geq 1-\alpha +\sum_{j \in J} a_j$ & $\forall J \subset [Dk] \& |J| = k_2$\\
$(2)$ $a_j + b_j \geq \frac{\alpha - \sum_{j'=1}^{j-1}a_{j'}}{k_2}$ & $\forall 1 \leq j \leq Dk$\\ 
$(3)$ $\sum_{j=1}^{Dk}  b_j \leq 1- \alpha$ & \\ 
$(4)$ $0 \leq \alpha, a_j, b_j \leq 1$ & $\forall 1 \leq j \leq Dk$ \\
$(5)$ $\beta \geq \sum_{j=1}^k a_j+b_j$ & \\ 
$(6)$ $a_j+b_j \geq a_{j+1}+b_{j+1}$ & $\forall 1 \leq j < Dk$ \\ \hline
\end{tabular} \vspace{0.2in}

Now we prove how to make $a^*$ variables monotone decreasing.  Suppose for some $j_1 < j_2$, we have $a^*_{j_1} = a^*_{j_2} - 2\delta$ for some positive $\delta$. We set both variables $a^*_{j_1}$ and $a^*_{j_2}$ to their average, i.e. increase $a^*_{j_1}$ by $\delta$, and decrease $a^*_{j_2}$ by $\delta$.

$$
a^*_{j_1} := a^*_{j_1} + \delta 
$$

$$
a^*_{j_2} := a^*_{j_2} - \delta
$$

We also decrease $b^*_{j_1}$ and increase $b^*_{j_2}$ by $\delta$: 

$$
b^*_{j_1} := b^*_{j_1} - \delta
$$

$$
b^*_{j_2} := b^*_{j_2} + \delta
$$

Now, we show that all constraints holds one by one. We note that it suffices to consider constraint $(1)$ only for set $J$ with maximum $\sum_{j \in J} a^*_j$. 
We can assume that this set $J$ with maximum $\sum_{j \in J} a^*_j$
cannot contain $j_1$ without having $j_2$ (either before of after the update) because $a^*_{j_1}$ is at most $a^*_{j_2}$ in both cases. 
Therefore, the right hand side of constraint $(1)$ is intact, and it still holds. 

Constraints $(2)$, $(5)$, and $(6)$ are all intact because $a^*_j+b^*_j$ is invariant in this operation for any $1 \leq j \leq Dk$.

Constraint $(3)$ holds because sum of $b^*$ variables remain the same. 
To prove constraint $(4)$ holds, it suffices to show that all variables stay in range $[0,1]$. It is evident for $a^*$ values since we are setting them to their average. For $b^*$ values, we first prove that $b^*_{j_1}$ stays non-negative. We note that $a^*_{j_1} + b^*_{j_1} \geq a^*_{j_2} + b^*_{j_2}$ using constraint $(6)$. We also have $a^*_{j_1} = a^*_{j_2}$, and $b^*_{j_2} \geq 0$. Therefore, $b^*_{j_1}$ cannot be negative. To prove that $b^*_{j_2}$ is (still) at most $1$, it suffices to  note that the sum of all  $b^*$ values is (still) at most $1- \alpha \leq 1$, and $b^*$ variables are all non-negative. Therefore all constraints are still valid after this operation. 
 
We prove that after a finite number of times (at most ${DK \choose 2}$) of applying this operation, we reach a feasible solution with monotone non-increasing sequence of $a^*$ values. 
If we start with $j_1=1$, and do this operation for any pair $(j_1,j_2)$ 
with $a^*_{j_1} < a^*_{j_2}$, after at most $Dk-1$ steps, we reach a solution in which $a^*_1 \geq a^*_j$ for any $1 < j \leq Dk$.
 We continue the same process by increasing $j_1$ one by one, and after at most ${DK \choose 2}$ updates, we reach a sorted sequence of $a^*$ values.  By monotonicity of $a^*$ values, we can simplify the linear program even further: 

\vspace{0.2in}
\begin{tabular}{|ll|} \hline $LP^{k,k_2}$ Minimize $\beta$ & \\ \hline 
Subject to: & \\
$(1)$ $\beta \geq 1-\alpha +\sum_{j=1}^{k_2} a_j$ &\\
$(2)$ $a_j + b_j \geq \frac{\alpha - \sum_{j'=1}^{j-1}a_{j'}}{k_2}$ & $\forall 1 \leq j \leq Dk$\\ 
$(3)$ $\sum_{j=1}^{Dk}  b_j \leq 1- \alpha$ & \\ 
$(4)$ $0 \leq \alpha, a_j, b_j \leq 1$ & $\forall 1 \leq j \leq Dk$ \\ 
$(5)$ $\beta \geq \sum_{j=1}^k a_j+b_j$ & \\ 
$(6)$ $a_j+b_j \geq a_{j+1}+b_{j+1}$ & $\forall 1 \leq j < Dk$ \\ 
$(7)$ $a_{j}  \geq a_{j+1}$ & $\forall 1 \leq j < Dk$ \\ \hline
\end{tabular} \vspace{0.2in}

Now we prove that we can assume that constraint $(2)$ is tight for all $1 \leq j \leq Dk$. At first, we prove by contradiction that the right hand side of constraint $(2)$ is non-negative. Let $j$ be the minimum index for which the right hand side of constraint $(2)$ is negative. 
We can set all $a^*$ and $b^*$ values to zero for any index greater than $j$, and also reduce $a^*_{j-1}$ by some amount to make this right hand side zero. All constraints hold, and we will have a solution in which the right hand side of constraint $(2)$ is always non-negative. Now we make constraint $(2)$ always tight as follows. 
Let $j_1$ be the maximum index in range $[1,Dk]$, for which constraint $(2)$ is loose by some $\delta > 0$. We update the variables as follows. 
If $b^*_{j_1}$ is positive, we reduce it to $\max\{b^*_{j_1} - \delta, 0\}$, and do not change any other variable. We note that in this case, all constraints still hold, and constraint $(2)$ for all indices $j_1, j_1+1, \cdots, Dk$ is tight.

If $b^*_{j_1}$ is zero, we decrease $a^*_{j_1}$ by $\delta$, and for any $j_2 > j_1$, we increase $a^*_{j_2}$ by $\frac{\delta}{k_2}(1-\frac{1}{k_2})^{j_2-j_1-1}$.
We prove that constraint $(2)$ for all indices $j_1, j_1+1, \cdots, Dk$ is tight, and all other constraints still hold after this update. 
By definition of $\delta$, constraint $(2)$ is tight for index $j_1$ after this update. Constraint $(2)$ was tight before the update for $j_2 > j_1$ because of the special choice of $j_1$. We prove that the right and left hand sides of constraint $(2)$ increased by the same amount for each $j_2 > j_1$. The left hand side increased by $\frac{\delta}{k_2}(1-\frac{1}{k_2})^{j_2 - j_1 -1}$.
We also know that the right hand side changed by: 

$$
\frac{\delta - \sum_{j'=j_1+1}^{j_2-1} \frac{\delta}{k_2}(1-\frac{1}{k_2})^{j' - j_1 -1}}{k_2} = \frac{\delta - \frac{\delta}{k_2} k_2 (1-(1-\frac{1}{k_2})^{j_2-j_1-1}) }{k_2} = \frac{\delta (1-\frac{1}{k_2})^{j_2-j_1-1}}{k_2}
$$ 

Therefore constraint $(2)$ is tight for all $j_2 > j_1$ after the update. 
Now we prove feasibility of the new solution. 
For constraint $(1)$, we note that all increments of $a^*$ variables is less than $\frac{\delta}{k_2} \sum_{r=0}^{\infty} (1-\frac{1}{k_2})^r = \delta$. We also note that $a^*_{j_1}$ is decreased by $\delta$. So the right hand side of constraint $(1)$ is decreased, and it remains feasible. We just showed that for $j_2 \geq j_1$, constraint $(2)$ is tight and therefore valid, and it is intact for smaller indices. Constraint $(3)$ is also intact. To prove constraint $(4)$, we should show that the new $a^*$ values are in range $[0,1]$. Since constraint $(2)$ is tight, and $\alpha$ is at most $1$, these new $a^*$ values are all at most $1$. Non-negativity of the right hand side of constraint $(2)$ implies that these new values are all non-negative. Constraint $(5)$ holds since its right hand side is only decreased. To prove constraint $(6)$, we note that the right hand sides of constraint $(2)$ is decreasing in $j$, and they are all tight for indices $\geq j_1$. Therefore constraint $(6)$ holds for $j \geq j_1$. For $j = j_1 -1$, it clearly holds since we are decreasing $a^*_{j_1}$, and consequently its right hand side. To prove constraint $(7)$, we note that the increments in $a^*$ values is decreasing as $j_2 > j_1$ increases. So constraint $(7)$ remains feasible for $j > j_1$. For $j=j_1$, we note that $b^*_{j_1}$ is zero. So using constraint $(6)$, we can prove constraint $(7)$ holds for $j=j_1$ which completes the feasibility proof. 
Each time, we make these updates, the index $j_1$ (the maximum index for which constraint $(2)$ is loose) reduces by at least $(1)$. Therefore after at most $Dk$ operations, we have an optimum solution with all three properties of Claim~\ref{claim:LP_optimum_solution_structure}.
\end{proof}

Using Claim~\ref{claim:LP_optimum_solution_structure}, we can
assume that the solution of $LP^{k,k_2}$ is lower bounded by the next $LP^{new, k, k_2}$. We focus on lower bounding the solution of $LP^{new, k, k_2}$ in the rest of the proof. 

\vspace{0.2in}
\begin{tabular}{|ll|} \hline $LP^{new, k,k_2}$ Minimize $\beta$ & \\ \hline 
Subject to: & \\
$(1)$ $\beta \geq 1-\alpha +\sum_{j=1}^{k_2} a_j$ &\\
$(2)$ $a_j + b_j = \frac{\alpha - \sum_{j'=1}^{j-1}a_{j'}}{k_2}$ & $\forall 1 \leq j \leq Dk$\\ 
$(3)$ $\sum_{j=1}^{Dk}  b_j \leq 1- \alpha$ & \\ 
$(4)$ $0 \leq \alpha, a_j, b_j \leq 1$ & $\forall 1 \leq j \leq Dk$ \\  
$(5)$ $a_j+b_j \geq a_{j+1}+b_{j+1}$ & $\forall 1 \leq j < Dk$ \\ 
$(6)$ $a_{j}  \geq a_{j+1}$ & $\forall 1 \leq j < Dk$ \\ \hline
\end{tabular} \vspace{0.2in}

We note that we eliminated one of the lower bounds on $\beta$. 
This only reduces the optimum solution of the linear program which is consistent with our approach. We also used Claim~\ref{claim:LP_optimum_solution_structure} to replace the inequality constraint $(2)$ with an equality constraint. We also removed the $c$ variables, and added the monotonicity constraint $(6)$. In the rest of the proof, we introduce some notation, and show some extra structure in the optimum solution of $LP^{new, k, k_2}$. This will help us lower bound the optimum solution by analyzing a system of two equations explicitly.  

We start with proving the extra structure. Let $\tau \defeq a_{k_2}$. 
We show that for any pair of indices $1 \leq j_1 < j_2 \leq k_2$, either $b_{j_1}$ is zero or $a_{j_2}$ is equal to $\tau$.
Let $j_1$ be the minimum index with $b_{j_1} > 0$. If $j_1$ is at least $k_2$, the claim holds clearly. So we consider $j_1 < k_2$. 
If $a_{j_1+1}$ is equal to $\tau$, by monotonicity of $a$ values, the claim is proved. So we define $\delta > 0$ to be $\min\{b_{j_1}, a_{j_1+1} - \tau\}$. 
We increase $a_{j_1}$, and $b_{j_1+1}$ by $\delta$, and $\delta(1-\frac{1}{k_2})$ respectively. We also decrease both of $a_{j_1+1}$, and $b_{j_1}$ by $\delta$.
After these changes, constraints $(1)$, $(2)$, $(3)$, $(4)$, and $(5)$ in $LP^{new,k,k_2}$ still hold. 
In particular, we made the changes in this special way to make sure that constraint $(2)$ still holds. Constraint $(5)$ also holds since the right hand side of constraint $(2)$ is decreasing in $j$. 
But the monotonicity constraint $(6)$ may be violated for $j=j_1+1$. This happens if the new $a_{j_1+1}$ is less than $a_{j_1+2}$. In this case, we swap the variables $a_{j_1+1}$ and $a_{j_1+2}$. We also change $b_{j_1+1}$, and $b_{j_1+2}$ in a way that constraint $(2)$ holds for both $j = j_1+1$, and $j=j_1+2$. Similarly, we have that all constraints $(1), (2), \cdots, (5)$ hold. But constraint $(6)$ may be violated for $j=j_1+2$. We continue doing this swap operation until constraint $(6)$ holds as well. This will happen after at most $k_2$ swap operations, since the new $a$ values are all at least $\tau$.
Finally, we reach a feasible solution for  $LP^{new,k,k_2}$ in which either one more $a$ variable is equal to $\tau$ or one more $b$ variable is set to zero. Therefore after at most $2k_2$ updates, for any pair of indices $1 \leq j_1 < j_2 \leq k_2$, we have that either $b_{j_1}$ is zero or $a_{j_2}$ is equal to $\tau$.

We also claim that for any $j > k_2$, we can assume either $a_j=\tau$, or $b_j=0$. Otherwise, suppose $j_1 > k_2$ is the smallest index for which $a_{j_1} < \tau$, and $b_{j_1} > 0$. We can increase $a_{j_1}$ by $\delta = \min\{\tau - a_{j_1}, b_{j_1}\}$, and decrease $b_{j_1}$ by $\delta$. We note that since $a_{j_1}+ b_{j_1}$ is invariant, the monotonicity constraints $(5)$ still holds. 
We need to prove constraint $(6)$ for $j=j_1-1$. It can be violated only if $a_{j_1-1}$ is less than $\tau$, and $b_{j_1-1}$ is non-negative which contradicts the choice of $j_1$. 
To prove feasibility, we only need to prove constraint $(2)$  for $j > j_1$. We make it hold by the following adjustments. We start by $j=j_1+1$, and increase it one by one. If constraint $(2)$ is loose by some $\epsilon$ for index $j$, we decrease $b_j$ by $\min\{b_j, \epsilon\}$. We also decrease $a_j$ by $\max\{\epsilon - b_j, 0\}$. The sum $a_j+b_j$ is reduced by $\epsilon$, and constraint $(2)$ now holds for $j$. It is also clear that sum of $b$ variables do not increase, and all other constraints still hold. With these adjustments for each $j > j_1$ in the increasing order, we know the solution is feasible. We also have all the ingredients to characterize the optimum solution of $LP^{new, k, k_2}$, and lower bound it.

 We can formalize this optimum solution in terms of a few parameters including $\tau, \beta, k,$ and $k_2$. 
 At this final stage of the proof, we conclude with two lower bounds (system of two equations) on $\beta$ and $1-\alpha$ in terms of these few parameters. 
 Let $t$ be the smallest index in range $1 \leq t \leq k_2$ with $b_t \neq 0$. If such an index does not exist, define $t$ to be $k_2$.  Using Claim~\ref{claim:LP_optimum_solution_structure}, we know that constraint $(2)$ is tight, and $b_j = 0$ for any $j < t$, we can inductively prove that $a_j = \frac{\alpha}{k_2} (\frac{k_2-1}{k_2})^{j-1}$ for any $j < t$. 
Consequently, we have that $\sum_{j=1}^{t-1} a_j$ is equal to $\alpha (1 - (\frac{k_2-1}{k_2})^{t-1}) \geq \alpha (1 - e^{-r})$ where $r$ is defined to be $\frac{t-1}{k_2}$. Therefore, constraint $(1)$ implies that:
$\beta \geq 1- \alpha + (1-e^{-r})\alpha + (1-r)k_2\tau$. This lower bound on $\beta$ is the first inequality we wanted to prove. To achieve the second inequality (lower bound on $1-\alpha$), we start by upper bounding the sum of $a$ variables. 

 We show that $\sum_{j=1}^t a_j \leq \alpha\left(1 - e^{-r} + \frac{2}{k_2}\right)$ as follows.
 Since $a$ variables are monotone and, constraint $(2)$ is tight for $j=1$, we have $a_t \leq a_1 = \frac{\alpha}{k_2}$. 
We also have that $\sum_{j=1}^{t-1} a_j$ is equal to $\alpha \left(1 - \left(1-\frac{1}{k_2}\right)^{t-1}\right)$. 
 We can upper bound $\left(1-\frac{1}{k_2}\right)^{k_2-1}$ by $e^{-1}$ as follows. 
We prove that  $\left(1-\frac{1}{k_2}\right)^{k_2-1}$ is a monotone decreasing sequence for $k_2 = 1, 2, \cdots$. 
 
 $$
 \frac{\left(1-\frac{1}{k_2}\right)^{k_2-1}}{\left(1-\frac{1}{k_2+1}\right)^{k_2}} =
 \frac{1}{1-\frac{1}{k_2+1}} \times \left(1-\frac{1}{k_2^2} \right)^{k_2-1} \geq
 \frac{k_2+1}{k_2} \times \left(1-\frac{k_2-1}{k_2^2} \right)  = 
 \frac{k_2^3+1}{k_2^3} > 1
 $$
 
 We also know that $\lim_{k_2 \to \infty} \left(1-\frac{1}{k_2}\right)^{k_2-1} = e^{-1}$. Therefore each term $\left(1-\frac{1}{k_2}\right)^{k_2-1}$ is at least $e^{-1}$. 
Therefore, we have that

 \begin{eqnarray*}
 \sum_{j=1}^t a_j 
 &\leq& 
 \alpha \left(1-\left(\frac{k_2-1}{k_2}\right)^{t-1}\right) + \frac{\alpha}{k_2}
\leq
\alpha \left(1-\left(\left(\frac{k_2-1}{k_2}\right)^{k_2}\right)^{r}\right) + \frac{\alpha}{k_2} \\
&\leq&
\alpha \left(1-\left(\left(\frac{k_2-1}{k_2}\right)^{k_2-1} \left(\frac{k_2-1}{k_2}\right)\right)^{r}\right) + \frac{\alpha}{k_2}
\leq
\alpha \left(1-e^{-r}\left(1-\frac{1}{k_2}\right)^r + \frac{1}{k_2}\right) \\
&\leq&    
\alpha \left(1-e^{-r}\left(1-\frac{r}{k_2}\right) + \frac{1}{k_2}\right) 
\leq
\alpha \left(1-e^{-r} + \frac{2}{k_2}\right)
 \end{eqnarray*}
 
 which yields the desired upper bound on $\sum_{j=1}^t a_j$.

 Since $a_j$ is at most $\tau$ for any $j > t$, constraint $(2)$ implies that $b_j 
 \geq \frac{\alpha - \alpha\left(1-e^{-r} + \frac{2}{k_2}\right) - (j-t-1)\tau}{k_2} - \tau$. 
 To simplify the calculation, let $\alpha'$ be $\alpha - \alpha\left(1-e^{-r} + \frac{2}{k_2}\right) = \alpha\left(e^{-r} - \frac{2}{k_2}\right)$. Summing up the lower bounds on $b$ values imply that: $1-\alpha \geq \sum_{j=1}^{Dk} b_j \geq \sum_{\ell=0}^{\min\{Dk-t-1,\ell^*\}} \left( \frac{\alpha' - \ell\tau}{k_2} - \tau \right)$ where $\ell^*$ is the greatest integer  less than or equal to $\frac{\alpha'}{\tau} - k_2$ (we set it to zero if $\frac{\alpha'}{\tau} - k_2$ is not positive). By this definition of $\ell^*$, we make sure that the right hand side summands are non-negative, and therefore they will not weaken the inequality.
 The rest of the analysis is done in two cases. If $\ell^* \leq Dk-t-1$, by computing the sum, we achieve following inequality: 

$$
1-\alpha \geq
\sum_{j=1}^{Dk} b_j \geq
  \frac{\ell^*+1}{k_2} \left(\alpha' - k_2\tau - \frac{\ell^*\tau}{2} \right) 
\geq \frac{\alpha' - k_2\tau}{k_2\tau} \left(\alpha' - k_2\tau - \frac{\alpha' - k_2 \tau}{2} \right) 
= \frac{\left(\alpha' - k_2 \tau \right)^2}{2k_2\tau}
 $$

We conclude that if $\beta^*$ is the solution of linear program $LP^{new,k,k_2}$, the following system of equations should have a solution with $\beta = \beta^*$: 

\begin{eqnarray}\label{eq:SLE1}
\beta &\geq& 1- \alpha + (1-e^{-r})\alpha + (1-r)\lambda   \\
1 - \alpha &\geq& \frac{\left(e^{-r}\alpha - \frac{2\alpha}{k_2} - \lambda \right)^2}{2\lambda} \nonumber
\end{eqnarray}
  
 where $\lambda$ is defined to be $k_2\tau$. We note that $\alpha, \lambda,$ and $r$ are the variables of the above system of two equations, and they should be in range $[0,1]$. We also note that $k_2$ is another variable which can be any positive integer. To simplify, we solve the following system of equations to eliminate $k_2$: 
 
\begin{eqnarray}\label{eq:SLE2}
\beta &\geq& 1- \alpha + (1-e^{-r})\alpha + (1-r)\lambda   \\
1 - \alpha &\geq& \frac{\left(e^{-r}\alpha - \lambda \right)^2}{2\lambda} \nonumber
\end{eqnarray}
 
It is easy to see that if system of equations~\ref{eq:SLE1} has a solution $(\beta_1, \alpha_1, \lambda_1, r_1, k_2)$, system of equations~\ref{eq:SLE2} has the following solution: $(\beta = \beta_1+\frac{2}{k_2}, \alpha = \alpha_1, \lambda = \lambda_1 + \frac{2}{k_2}, r = r_1)$. Therefore it suffices to lower bound $\beta$ in system of equations~\ref{eq:SLE2}. Because the same lower bound plus the term  $\frac{2}{k_2} = O(\frac{1}{k})$ holds for $\beta$ in system of equations~\ref{eq:SLE1}. 

By computing the partial derivatives, and considering boundary values, one can find the minimum $\beta$ for which the system of equations~\ref{eq:SLE2} has a valid solution.  Its minimum occurs when $r$ is zero, and the second inequality is tight. Therefore we have $\alpha= \sqrt{\lambda(2-\lambda)}$. We conclude that $\beta$ is the minimum of $1-\sqrt{\lambda(2-\lambda)} + \lambda$ which is equal to $1 - \sqrt{(1-\sqrt{\frac{1}{2}})(1+\sqrt{\frac{1}{2}})}+ (1-\sqrt{\frac{1}{2}}) = 2 - 2\sqrt{\frac{1}{2}} = 2 -\sqrt{2} \approx 0.5857$ and occurs at $\lambda = 1 - \sqrt{\frac{1}{2}} \approx 0.2928$.

In the other case, $\ell^*$ is greater than $Dk-t-1$, and therefore $\alpha'$ is at least $Dk\tau$ by definition of $\ell^*$. Since $t \leq k_2$, we can write the following lower bound on $1-\alpha$: 

$$
1-\alpha \geq
\sum_{j=1}^{2k} b_j \geq
  \frac{Dk-t}{k_2} \left(\alpha' - k_2\tau - \frac{(Dk-t-1)\tau}{2} \right) 
\geq \frac{D-1}{2} \left(\alpha' - k_2\tau\right) 
 $$

So we have a slightly different set of two inequalities in this case to lower bound $\beta$.  

\begin{eqnarray}\label{eq:SLE3}
\beta &\geq& 1- \alpha + (1-e^{-r})\alpha + (1-r)\lambda   \\
1 - \alpha &\geq& \frac{D-1}{2}\left(e^{-r}\alpha - \lambda \right) \nonumber
\end{eqnarray}

 It is evident that both inequalities should be tight to minimize $\beta$, and therefore $\alpha$ is equal to $\frac{1+\lambda}{1+D'e^{-r}}$ where $D'$ is $\frac{D-1}{2}$. So we can write $\beta$ as a function of just $\lambda$ and $r$: 
 
 $$
 \beta(\lambda,r) = 1 - \frac{1+\lambda}{1+D'e^{-r}}e^{-r} +(1-r)\lambda
 $$  
 
To minimize $\beta$, either $\lambda$ should be at one of its boundary values $\{0,1\}$, or the partial derivative $\frac{\partial \beta}{\partial \lambda}$ should be zero.
 For $\lambda=1$, $\beta$ cannot be less than $2 - r -e^{-r} \geq 1- \frac{1}{e} > 2-\sqrt{2}$. For $\lambda=0$, we have $1-\alpha \geq D'\alpha' \geq D'\alpha$, so $\alpha$ is at most $\frac{1}{1+D'} = \sqrt{2} -1$, and therefore $\beta$ is at least $1-\alpha \geq 2-\sqrt{2}$. The only case to consider is when $\frac{\partial \beta}{\partial \lambda} = 0$ which means $\frac{e^{-r}}{1+D'e^{-r}}$ should be equal to $1-r$ with a unique solution $r^* = 0.71 \pm 0.001$. Therefore $\beta$ is equal to $1-(1-r^*)(1+\lambda-\lambda) = r^* > 2-\sqrt{2}$. We conclude that any feasible solution of linear program $LP^{k,k_2}$ 
has $\beta$  at least $2-\sqrt{2} - O(\frac{1}{k})$ which completes the proof.
\end{proof}

\end{proof}

We show in the following Theorem that the $(2-\sqrt{2}) \approx 0.585$ lower bound on the approximation ratio of the core-sets that $\greedy$ finds is tight even if we allow the core-sets to be significantly large.

\begin{theorem} \label{thm:585tight}
For any $\epsilon > 0$, and any core-set size $k' \geq k$, there are instances of monotone submodular maximization problem with cardinality constraint $k$ for which $\greedy$ is at most a $(2-\sqrt{2}+O(\epsilon))$-approximate randomized composable core-set even if each item is sent to $C  \leq \sqrt{\epsilon m}$ machines.
\end{theorem}

\begin{proof}
Let $\NN$ be a subfamily of subsets of a universe ground set $\UU$, and let
$f:2^\NN\rightarrow R$ be defined as follows: for any set $S \subseteq \NN$, $f(S) = |\cup_{A \in S} A|$. 
Note that $f$ is a coverage function, and thus it is a monotone submodular function. 

Our hardness instance is inspired by the solution of Equation~\ref{eq:SLE2} in the proof of Theorem~\ref{thm:two_minus_sqrt_two_core-set}. We remind the solution $\left(\beta = 2 - \sqrt{2}, \alpha = \sqrt{\frac{1}{2}}, \lambda = k_2 \tau = 1 - \sqrt{\frac{1}{2}}\right)$. Let $k_2=k-1$, and $k_3=\lceil \frac{\alpha k_2}{\lambda}\rceil$.
For each $1 \leq i \leq k_2$, and $1 \leq j \leq k_3$, let $B_{i,j}$ be a set of $\lceil \frac{\alpha}{\epsilon}\rceil$ arbitrary elements in universe.  Let $B'$ be a set of  $(1+\epsilon)k k_3 \lceil \frac{1-\alpha}{\epsilon}\rceil$ arbitrary items of $\N$. We assume sets $B_{i,j}$s, and set $B'$ are all pairwise disjoint. Define set $R_i$ to be the union $\cup_{j=1}^{k_3} B_{i,j}$ (row $i$ of matrix $B$ of sets) for any $1 \leq i \leq k_2$. 
Similarly define set $C_j$ to be the union $\cup_{i=1}^{k_2} B_{i,j}$ (column $j$ of matrix $B$ of sets) for any $1 \leq j \leq k_3$. Define set $X_j$ to be union of $C_j$, and $\max\{1,\left( k_3  - k_2 - (j-1)\right)  \lceil \frac{\alpha}{\epsilon}\rceil + 1\}$ arbitrary elements of $B'$ such that sets $\{X_j\}_{j=1}^{k_3}$ are pairwise disjoint. We note that there are enough elements in $B'$ to make these sets disjoint because of Equation~\ref{eq:SLE2}.  Since we want sets $X_j$s to be disjoint we need at least $k_3+\sum_{j=1}^{k_3-k_2+1} \left( k_3  - k_2 - (j-1)\right)  \lceil \frac{\alpha}{\epsilon}\rceil$ items in $B'$. Therefore $\vert B'\vert$ suffices to be as large as $k_3+\frac{(k_3-k_2)(k_3-k_2+1)}{2}\times \lceil \frac{\alpha}{\epsilon}\rceil \leq \frac{(\alpha - \lambda)^2}{2} \frac{\alpha}{\lambda^2 \epsilon} (1+\epsilon)$. Using Equation~\ref{eq:SLE2}, we have  $1-\alpha \geq \frac{(\alpha - \lambda)^2}{2\lambda}$ which completes the proof of the lower bound claim on size of $B'$.   
We also define $k'$ singleton disjoint sets $\{Z_{\ell}\}_{\ell=1}^{k'}$ where $k'$ is the the number of items $\greedy$ is allowed to return. Each $Z_{\ell}$ contains one element of universe, and is disjoint from all other sets defined above. 

The hardness instance consists of one copy of set $B'$, one copy of set $R_i$ (for each $1 \leq i \leq k_2$), $10m\ln(mk)$ copies of set $X_j$ (for each $1 \leq j \leq k_3$), and $10m\ln(mk)$ copies of set $Z_{\ell}$ (for each $1 \leq \ell \leq k'$)  where $m$ is the number of machines.

First, we prove that with high probability, each machine
has at least one copy of $X_j$ for each $1 \leq j \leq k_3$. We note that for each $j$, and each machine $M$, at least $10\ln(mk)$ copies of $X_j$ is sent to $M$, and using Chernoff bound~\ref{lem:ChernoffBound}, with probability at least $1 - \frac{1}{m^2k^2}$ at least one copy of $X_j$ is sent to $M$. Using Union bound, with probability at least $1 - \frac{1}{mk}$, for each machine is receiving at least one copy of $X_j$ for any $j$. Similarly we know that each machine has at least one copy of set $Z_{\ell}$ for any $1 \leq \ell \leq k'$.

We also note that for any $1 \leq i \leq k_2$, with probability at least $1-\epsilon$, none of the at most $C$ copies of set $R_i$ shares a machine with one copy of $B'$. There are at most $C$ copies of $B'$, and $C$ copies of $R_i$, and the probability that two sets are sent to the same machine is $\frac{1}{m}$. So with probability at most $\frac{C^2}{m} \leq \epsilon$, one copy of $R_i$ and one copy of $B'$ are sent to the same machine.

In a machine $M$ that does not have any copy of  $B'$, and has at least one copy of $X_j$ (for any $1 \leq j \leq k_3$), $\greedy$ does not choose any copy of $R_i$. We prove this by contradiction. Suppose at iteration $t$, $\greedy$ chooses some set $R_i$ in this machine. Sets $X_1, X_2, \cdots, X_{t-1}$ are the first selected sets by $\greedy$ since $X_j$ sets are disjoint and monotone decreasing in size. We claim that the marginal gain of set $X_t$ is larger than the marginal gain of $R_i$ which shows the contradiction. Marginal gain of selecting $X_t$ is $k_2 \lceil \frac{\alpha}{\epsilon}\rceil + \left( k_3  -k_2 -  (t-1)\right)  \lceil \frac{\alpha}{\epsilon}\rceil + 1$. We also know that marginal gain of selecting any row set $R_i$ is equal to $\left(\lceil \frac{\alpha k_2}{\lambda}\rceil - (t-1)\right) \lceil \frac{\alpha}{\epsilon}\rceil$.  Clearly the marginal gain of adding set $X_t$ is larger which completes the proof of the claim. So $\greedy$ selects all sets $\{X_j\}_{j=1}^{k_3}$ at first.
We note that after choosing all these sets, marginal gain of adding any row set $R_i$ is zero, and therefore they will not be selected because $\greedy$ can always choose a set $Z_{\ell}$ to get a marginal value of at least $1$. 

So the expected value of best size $k$ subset of selected items (sets) is at most $\epsilon k_2  \vert R_i\vert +  k + \vert B' \vert + k \vert C_j \vert$ for any $i,j$. 
Because each row set $R_i$ is selected with probability at most $\epsilon$, and they all have the same size. There are at most $k$ column sets present in any set of $k$ selected sets. We also note that column sets $C_j$s have the same size. Set $B'$ is also counted in the computation, and at most $k$ singleton sets $Z_{\ell}$s should be counted. Therefore the maximum value of size $k$ subsets of selected sets is at most: 

\begin{eqnarray*}
\epsilon k_2  \vert R_i\vert +  k + \vert B' \vert + k \vert C_j \vert 
&\leq & 
 k k_3 \left( \epsilon \lceil \frac{\alpha}{\epsilon}\rceil + 1 + (1+\epsilon) \lceil \frac{1-\alpha}{\epsilon}\rceil  \right) + k^2 \lceil \frac{\alpha}{\epsilon}\rceil \\
&\leq & k_2k_3 (1+\frac{2}{k})(1+O(\epsilon)) (\alpha + 1 + \frac{1-\alpha}{\epsilon} + \frac{\lambda}{\epsilon}) \\
&\leq & k_2k_3 (1+O(\epsilon)) \frac{1-\alpha + \lambda}{\epsilon}
\end{eqnarray*}

where the last inequality holds for $k = \Omega(\frac{1}{\epsilon})$. On the other hand, the optimum solution consists of $k_2 = k-1$ row sets $R_i$s, and set $B'$ with value at least: 

$$
k_2 k_3 \lceil \frac{\alpha}{\epsilon}\rceil +  k_2k_3 \lceil \frac{1-\alpha}{\epsilon}\rceil \geq \frac{k_2k_3}{\epsilon}
$$

We conclude that the expected value of maximum value size $k$ subset of selected sets is upper bounded by $\lambda + 1-\alpha + O(\epsilon) = 2 - \sqrt{2} + O(\epsilon)$ times the optimum solution $f(\opt)$.
 \end{proof}

\subsection{Improved Distributed Approximation Algorithm}\label{subsec:PseudoGreedy}
We remind that in the first phase, each machine $1 \leq i \leq m$ runs algorithm $\greedy$ on set $T_i$ with $k'=Dk$ (where $D$ is $2\sqrt{2}+1$).
By Theorem~\ref{thm:two_minus_sqrt_two_core-set}, there exists a size $k$ subset of selected items $\cup_{i=1}^m S_i$ with expected value at least $0.585f(\opt)$, but we do not know how to find this set efficiently. 
If we  apply algorithm $\greedy$ again on $\cup_{i=1}^m S_i$ to select $k$ items in total, 
we achieve a distributed approximation factor of $(1-\frac{1}{e})\times 0.585 \approx 0.37$. 
In the following, we present a post-processing algorithm $\pseudogreedy$ that achieves an overall distributed approximation factor better than $1/2$. In particular, after the first phase, we show how to find a size $k$ subset of the union of selected items $\cup_{i=1}^m S_i$ with expected value at least $(0.545-o(1))f(\opt)$. 

Algorithm $\pseudogreedy$ proceeds as follows: it first computes a family of candidate solutions of size $k+O(1)$, and keeps the one candidate solution $V$ with the maximum value. It then lets $S$
to be a random size $k$ subset of $V$, and returns $S$ as the solution. 
These candidate solutions, denoted by $S_{k'_2, I}$ (for $1\le k'_2\le k$ and $4I\subseteq \{1,\cdots, 8\}$) in Algorithm~\ref{alg:54-approximate}, are  computed as follows: 
We first  enumerate all $k$ possible values of $1 \leq k'_2 \leq k$ (this notation is used to be consistent with the proof of Theorem~\ref{thm:two_minus_sqrt_two_core-set}). Then, by letting $k'_1=k-k'_2$, and $k_3=32 \lceil \frac{k'_2}{128}\rceil$, we partition the first $8k_3$ items in set $S_1$\footnote{ We choose any machine instead of machine 1, and since the clustering is done at random, the analysis goes through.
} into eight subsets  $\{A_{i'}\}_{i'=1}^8$, each of size $k_3$. 
The next step of the algorithm proceeds as follows: 
for any $I \subseteq \{1, 2, \cdots, 8\}$ with $\lvert I\rvert \leq 4 + \frac{k'_1}{k_3}$, we initialize set $S_{k'_2,I}$ with the union of all sets $A_{i'}$ where $i' \in I$. Then for $k'_1 + (4- \lvert I \rvert)k_3$ iterations, we search all items in $\cup_{i=1}^m S_i$, and insert the item with the maximum marginal value to $S_{k_2,I}$. Roughly speaking, by starting from all these subsets, we ensure that the selected set hits enough number of items in $\opt$.
The upper bound we enforce on $\lvert I \rvert$ is to make sure that the number of iterations at this step is non-negative, i.e. $k'_1 + (4- \lvert I \rvert)k_3 \geq 0$. 
Finally 
we define $V$ to be the set $S_{k'_2,I}$ with the maximum $f$ value, and return a random subset of size $k$ of $V$ as the output set $S$.

\vspace{0.07 in}
\begin{minipage}[t]{0.8\columnwidth}
\fbox{
\begin{algorithm}[H]\label{alg:54-approximate}
\textbf{Input:} A collection of $m$ subsets $\{S_1,\ldots, S_m\}$.\\[0.6ex]
 \textbf{Output:} Set $S \subset \cup_{i=1}^m S_{i}$ with $|S| \leq k$.\\[0.6ex]
$V \gets \emptyset$\;
\ForAll{$1 \leq k'_2 \leq k$}{
	$k_3 \gets 32\lceil \frac{k'_2}{128}\rceil$\; $k'_1 \gets k-k'_2$\;
	Partition the first $8k_3$ items of $S_1$ into $8$  sets $\{A_{i'}\}_{i'=1}^{8}$ each of size $k_3$\;
	\ForAll{$I \subseteq \{1, 2, \cdots, 8\}$ with $\vert I \vert \leq 4 + \frac{k'_1}{k_3}$}{
		$S_{k'_2,I} \gets \cup_{i' \in I} A_{i'}$\;
		\For{$k'_1 + (4 - \vert I \vert)k_3$ times}{
			Find $\argmax_{x \in \cup_{i=1}^m S_i} \Delta(x,S_{k'_2,I})$, and insert it to $S_{k'_2,I}$\;
		}
		{\bf if} {$f(S_{k_2,I}) > f(V)$} {\bf then} $V \gets S_{k_2,I}$\;
	}
}
$S \gets$ a random size $k$ subset of $V$\;
 \caption{Algorithm $\pseudogreedy$}
\end{algorithm}
}
\end{minipage}

\begin{theorem}\label{thm:54approximate}
Algorithm $\pseudogreedy$ returns a subset $S$ with size at most $k$, and expected value at least $(0.545-O(\frac{1}{k} +\frac{\ln(C)}{C}))f(\opt)$. 
\end{theorem}

\begin{proof}
Algorithm $\pseudogreedy$ returns set $S$, and we aim to prove that $\E[f(S)]$ is at at least $0.545(1-O(\frac{1}{k} + O(\frac{\ln(C)}{C}))) f(\opt)$ as follows. 
We prove that $\E[f(S)]$ is at at least $(1-O(\frac{1}{k} + O(\frac{\ln(C)}{C}))) f(\opt)$
 times the solution of the following linear program $LP^r$ for some $0 \leq r  \leq 1$.  This is a constant size linear program, and therefore we can find its optimum solution for any $r$. Although there are infinite number of choices for $r$, and we cannot find the optimum solution of all of them numerically, we can discretize the interval $[0,1]$, and consider a constant number of cases for $r$. In particular, we know that $r \in [\frac{d}{160}, \frac{d+1}{160}]$ for some $1 \leq d < 160$. In each of these cases, we have an accurate estimate of $r$, and based on these estimates, we change $LP^r$ in a way that does not violate feasibility of any solution (we should round up or down $r$ in a way that maintains feasibility of constraint $(1)$). Solving $LP^r$ for all these $160$ cases, we imply that its optimum solution is always above $0.545$ with the minimum happening for $r \approx \frac{127}{160}$. 
Therefore it suffices to show that for some $r$, the solution of linear program $LP^r$ times $(1-O(\frac{1}{k} + \frac{\ln(C)}{C}))$ is a lower bound on $\frac{\E[f(S)]}{f(\opt)}$.

\vspace{0.2in}
\begin{tabular}{|ll|} \hline $LP^r$ Minimize $\beta$ & \\ \hline 
Subject to: & \\
$(1)$ $\beta \geq (1-e^{-1-\frac{4-|I|}{4} \times \frac{r}{1-r}})(1-\alpha - \sum_{i' \in I} \sum_{\ell \in B_{i'}} b'_{\ell}) +\sum_{i' \in I} \sum_{\ell \in B_{i'}} a'_{\ell} + b'_{\ell} + c'_{\ell}$ & \\
 $\forall I \subset \{1, 2, \cdots, 8\} \& |I| \leq 4 + \frac{4(1-r)}{r}$ &\\
$(2)$ $a'_{\ell} + b'_{\ell} + c'_{\ell} \geq \frac{\alpha - \sum_{\ell'=1}^{\ell} a'_{\ell'}}{128}$ & $\forall 1 \leq \ell \leq 256$\\ 
$(3)$ $\sum_{\ell=1}^{256}  b'_{\ell} \leq 1- \alpha$ & \\ 
$(4)$ $0 \leq a'_{\ell}, b'_{\ell}, c'_{\ell}$ & $\forall 1 \leq \ell \leq 256$ \\ 
$(5)$ $0 \leq \alpha \leq 1$ &  \\ \hline
\end{tabular}
 \vspace{0.2in}

We need to borrow many notations from proof of Theorem~\ref{thm:two_minus_sqrt_two_core-set}. 
We remind that $T_1$, and $S_1$ are the sets of items sent to the first machine and the set of items selected in this machine respectively, i.e.,  $S_1$ is equal to $\greedy(T_1)$. We also remind that $D = 2\sqrt{2}+1$, and $k'=Dk$.
Note that we can assume $\vert S_1\vert = Dk$, since if there are less than $Dk$ items in $T_1$, WLOG we can assume the algorithm returns some extra dummy items just for the sake of analysis. 
Consider an item $x\in \opt$. We say that $x$ {\em survives from machine $1$}, if, when we send $x$ to machine $1$ in addition to items of $T_1$, algorithm $\greedy$ would choose this item $x$ in its output of size $k'$, i.e., if $x\in \greedy(T_1 \cup \{x\})\}$.
For the sake of analysis, we partition the optimum solution into two sets as follows:
let $\optMS$ be the set of items in the optimum solution that would survive the first machine, i.e., $\optMS\defeq \{x \vert x \in \opt \cap \greedy(T_1 \cup \{x\})\}$. 
Let $\optMNS \defeq \opt \setminus \optMS$, and
$k_1\defeq\vert \optMS\vert$, and $k_2\defeq\vert \optMNS\vert$ (note that $k_1+k_2 = k$).  
We also define $\optMS'\defeq \optMS  \cap \opt^S$, and $\opt' \defeq \optMS' \cup \optMNS$ where $\opt^S$ is defined in Subsection~\ref{sec:notation}. 
Using Lemma~\ref{lem:random_machine_selected_vs_total_selected} (restated here), we know that set $\E[f(\opt')]$ is almost equal to
$f(\opt)$. Therefore $\opt'$  can be used as the benchmark instead of $\opt$. In other words, assuming the following lemma, it suffices to prove that $\frac{f((S)}{f(\opt')}$ is at least $(1-O(\frac{1}{k}))$ times the solution of $LP^r$ for some $r$.

\begin{lemma}
The expected value of set $\opt'$, $\E[f(\opt')]$, is at least $\left(1-O\left(\frac{\ln(C)}{C}\right)\right) f(\opt)$.
\end{lemma}

Algorithm $\pseudogreedy$ picks a random size $k$ subset of set $V$, and returns it as the output set $S$. We prove that $V$ does not have more than $k+O(1)$ items. We note that $V$ is equal to set $S_{k'_2, I}$ for some $k'_2$ and $I$. Therefore $\lvert V \rvert$ is equal to  $k_3\vert I \vert + k'_1 + (4 - \vert I \vert)k_3 = k'_1 + 4k_3 \leq k'_1 + k'_2  + 128 = k+128$. 
Let $\sigma$ be an arbitrary permutation on items of $V$. 
We have that $f(V) = \sum_{x \in V} \Delta(x, \sigma^x)$ where $\sigma^x$ is the set of items in $V$ that appear before $x$ in permutation $\sigma$.
We also have that $f(S) = \sum_{x \in S} \Delta(x, \sigma^x \cap S)$. 
By submodularity, we know that $\sum_{x \in S} \Delta(x,\sigma^x \cap S)$ is at least $\sum_{x \in S} \Delta(x,\sigma^x)$. 
Therefore the ratio $\frac{\E[f(S)]}{f(V)}$ is at least $\frac{\sum_{x \in V} Pr[x \in S] \Delta(x,\sigma^x)}{\sum_{x \in V} \Delta(x,\sigma^x)}$ where the expectation is taken over the uniform distribution of all size $k$ subsets of $V$. We note that $Pr[x \in S]$ for each $x \in V$ is $\frac{k}{\lvert V \rvert} \geq 1- O(\frac{1}{k})$. We conclude that $\E[f(S)]$ is at least $(1 - O(\frac{1}{k}))f(V)$. Applying Lemma~\ref{lem:random_machine_selected_vs_total_selected}, we know that 
to complete the proof of Theorem~\ref{thm:54approximate},
it is sufficient to show $\frac{f(V)}{f(\opt')}$ is lower bounded by the solution of $LP^r$ for some $0 \leq r \leq 1$.

 Since algorithm $\pseudogreedy$ enumerates all $k$ possible values of $k'_2$, in one of these trials, $k'_2$ is equal to $k_2=\vert \optMNS \vert$. From now on, we focus on this specific value of $k'_2$. We remind that  $k_3=32 \lceil \frac{k'_2}{128}\rceil$. 
 Let $I'$ be the set of indices that maximizes $S_{k_2,I}$, i.e., 
 $I' \defeq \argmax_{I \subseteq \{1, 2, \cdots, 8\} \& \vert I \vert \leq 4 + \frac{k_1}{k_3}} f(S_{k_2,I})$. Define $V'$ be the set $S_{k_2,I'}$. 
By definition, we have $f(V') \leq f(V)$. So it suffices to show that 
 $f(V')$ is at least $f(\opt')$ times the solution of $LP^r$ for some $0 \leq r \leq 1$ as follows. We define $r$ to be $\frac{4k_3}{4k_3 + k_1}$. 

Before elaborating on how the linear program $LP^r$ lower bounds the ratio $\frac{f(V)}{f(\opt')}$, we need to explain this linear program in more details. Linear program $LP^r$ has $2+3\times 256 = 770$ variables: $\alpha$, $\beta$, $\{a'_{\ell}\}_{\ell=1}^{256}$, $\{b'_{\ell}\}_{\ell=1}^{256}$, $\{c'_{\ell}\}_{\ell=1}^{256}$. The objective function of $LP^r$ is variable $\beta$. Set $B_i$ of indices is defined to be $\{32(i-1)+1, 32(i-1)+2, \cdots, 32i\}$ for any $1 \leq i \leq 8$. We note that $r$, and $\lvert I \rvert$ are not variables, and therefore constraint $(1)$ like all other constraints is a linear inequality.

At first, we remind the feasible solution of $LP^{k,k_2}$ constructed in the proof of Lemma~\ref{lem:LP_lower_bound_core-set}, and then 
construct a feasible solution for $LP^r$  with $\beta$ equal to  
$\frac{f(V')}{f(\opt')}$ as follows.
Fix a permutation $\pi$ on the items of $\opt'$
such that every item of $\optMS'$ appears before every item of $\optMNS$ in $\pi$. In other words, $\pi$ is an arbitrary permutation on items of $\optMS'$ followed by an arbitrary permutation on items of $\optMNS$. 
For any item $x$ in $\opt'$, define $\pi^x$ to be the set of items in $\opt'$
that appear prior to $x$ in permutation $\pi$.	
We set $\alpha$ to be $\frac{\sum_{x \in \optMNS} \Delta(x,\pi^x)}{f(\opt')}$. 
For any $1 \leq j \leq 8k_3$, we define set $S^j$ to be the first $j$ items of $S_1$.
Now we are ready to define variables $\{a_j, b_j, c_j\}_{j=1}^{8k_3}$ as follows.

$$
a_j:=\frac{\sum_{x \in \optMNS} \Delta(x,\pi^x \cup S^{j-1}) - \Delta(x,\pi^x \cup S^{j})}{f(\opt')}
, 
b_j:=\frac{\sum_{x \in \optMS'} \Delta(x,\pi^x \cup S^{j-1}) - \Delta(x,\pi^x \cup S^{j})}{f(\opt')}
$$
 
We let the variable $c_j$ be the marginal gain of item $y_j$ divided by $f(\opt')$ minus $a_j+b_j$ where $y_j$ is the $j$th item selected in $S_1$, i.e., $c_j := \frac{\Delta(y_j,S^{j-1})}{f(\opt')} - a_j - b_j$.

We are ready to present a feasible solution for $LP^r$. We set $\beta$ to be $\frac{f(V')}{f(\opt')}$. We keep the same value for $\alpha = \frac{\sum_{x \in \optMNS} \Delta(x,\pi^x)}{f(\opt')}$. For any $1 \leq \ell \leq 256$, we define:

\begin{eqnarray*}
a'_{\ell} &\defeq& \sum_{j=\lceil \frac{k_2}{128}\rceil (\ell-1) + 1}^{\lceil \frac{k_2}{128}\rceil \ell}
 a_j \\
b'_{\ell} &\defeq& \sum_{j=\lceil \frac{k_2}{128}\rceil (\ell-1) + 1}^{\lceil \frac{k_2}{128}\rceil \ell}
 b_j \\
c'_{\ell} &\defeq& \sum_{j=\lceil \frac{k_2}{128}\rceil (\ell-1) + 1}^{\lceil \frac{k_2}{128}\rceil \ell}
 c_j 
\end{eqnarray*}

It suffices to prove feasibility of this solution for linear program $LP^r$. 
Using Lemma~\ref{lem:LP_lower_bound_core-set}, we know   $a, b,$ and $c$ variables with $\alpha$, and some choice of $\beta$ form a feasible solution for $LP^{k,k_2}$. Therefore, we can use constraints of $LP^{k,k_2}$ to prove feasibility of this solution for $LP^r$. Using constraint $(4)$ of $LP^{k,k_2}$, variable $\alpha$ is in range $[0,1]$ which proves constraint $(5)$ of $LP^r$. We also know that $a, b,$ and $c$ variables are non-negative which implies that by definition $a', b',$ and $c'$ variables are non-negative (constraint $(4)$ of $LP^r$). Constraint $(3)$ of $LP^r$ holds using constraint $(3)$ of $LP^{k,k_2}$ and by definition of $b'$. We can also prove constraint $(2)$ of $LP^r$ using the same constraint in $LP^{k,k_2}$ as follows. For each $\lceil \frac{k_2}{128}\rceil (\ell-1) + 1 \leq j  \leq \lceil \frac{k_2}{128}\rceil \ell$, we know that $a_j+b_j+c_j$ is at least $\frac{\alpha - \sum_{j'=1}^{j-1} a_{j'}}{k_2}$. Summing up this inequality for all $\lceil \frac{k_2}{128}\rceil (\ell-1) + 1 \leq j  \leq \lceil \frac{k_2}{128}\rceil \ell$, and by definition of $a', b',$ and $c'$, we have that: 

\begin{eqnarray*}
a'_{\ell} + b'_{\ell} + c'_{\ell} &\geq & \sum_{j=\lceil \frac{k_2}{128}\rceil (\ell-1) + 1}^{\lceil \frac{k_2}{128}\rceil \ell}
\frac{\alpha-\sum_{j'=1}^{j-1} a_{j'}}{k_2}
\geq 
\sum_{j=\lceil \frac{k_2}{128}\rceil (\ell-1) + 1}^{\lceil \frac{k_2}{128}\rceil \ell}
\frac{\alpha-\sum_{j'=1}^{\lceil \frac{k_2}{128}\rceil \ell} a_{j'}}{k_2} \\
&=&
\sum_{j=\lceil \frac{k_2}{128}\rceil (\ell-1) + 1}^{\lceil \frac{k_2}{128}\rceil \ell}
\frac{\alpha-\sum_{\ell'=1}^{\ell} a'_{\ell'}}{k_2} 
=
\lceil \frac{k_2}{128}\rceil \frac{\alpha-\sum_{\ell'=1}^{\ell} a'_{\ell'}}{k_2} 
\geq
\frac{\alpha-\sum_{\ell'=1}^{\ell} a'_{\ell'}}{128}
\end{eqnarray*}

We note that the last inequality holds assuming the numerator is non-negative. 
In case, the numerator is negative, the constraint $(2)$ holds using non-negativity of $a', b',$ and $c'$ variables. So we just need to prove constraint $(1)$ (which is the most important constraint) of $LP^r$. We remind that $\beta$ is defined to be  
$\frac{f(V')}{f(\opt')} = \max_{I \subseteq \{1, 2, \cdots, 8\} \& \vert I \vert \leq 4 + \frac{k_1}{k_3}} \frac{f(S_{k_2,I})}{f(\opt')}$. 
So for every $I \subseteq \{1, 2, \cdots, 8\}$ with size at most $4 + \frac{k_1}{k_3}$,  it suffices to prove that: 

$$
\frac{f(S_{k_2,I})}{f(\opt')} \geq 
\left(1-e^{-1-\frac{4- \vert I \vert}{4} \times \frac{r}{1-r}}\right)\left(1-\alpha - \sum_{i' \in I} \sum_{\ell \in B_{i'}} b'_{\ell}\right) +\sum_{i' \in I} \sum_{\ell \in B_{i'}} \left(a'_{\ell} + b'_{\ell} + c'_{\ell}\right)
$$

We note that  $S_{k_2,I}$ consists of two sets of items: 
\begin{itemize}
\item
Set $S(I)$ with $\vert I \vert k_3$ items corresponding to sets $\{B_{i'}\}_{i' \in I}$ which are added in the first phase.
In other words, $S(I)$ is $\cup_{i' \in I} B_{i'} = \cup_{i' \in I} \cup_{j = (i-1)k_3+1}^{ik_3} \{y_j\}$ where $y_j$ is the $j$th item of $S_1$.  
\item $k_1 + (4 - \vert I \vert)k_3$ items added greedily in the second phase.
\end{itemize}

Similar to proof of Claim~\ref{claim:LP_feasibility}, we define $S'(I)$ to be $\optMS' \cup S(I)$. To be consistent in notation, we define $J$ to be the indices of items in $S(I)$, i.e. $J \defeq \{j | y_j \in S(I)\}$. Using the same argument in proof of Claim~\ref{claim:LP_feasibility}, we know that:

\begin{eqnarray}\label{eq:S'I}
f(S'(I)) &\geq& (1-\alpha)f(\opt') + \sum_{j \in J} (a_j+c_j)f(\opt') \nonumber \\
&=& (1-\alpha)f(\opt') + \sum_{i' \in I} \sum_{\ell \in B_{i'}} (a'_{\ell}+c'_{\ell})f(\opt')
\end{eqnarray}

This means that there are $\vert \optMS' \vert \leq k_1$ items that can be added to $S(I)$ in the second phase, and increase its value to $f(S'(I))$. Since we are adding $k_1 + (4 - \vert I \vert)k_3$ items greedily in the second phase, our final value $f(S_{k_2,I})$ is at least:

\begin{eqnarray}\label{eq:Sk2I}
f(S_{k_2,I}) &\geq& f(S(I)) + (1-e^{-\frac{k_1 + (4 - \vert I \vert)k_3}{k_1}}) (f(S'(I)) - f(S(I))) \nonumber \\
&=& f(S(I)) + (1-e^{-1 - \frac{(4 - \vert I \vert)}{4} \times \frac{r}{1-r}}) (f(S'(I)) - f(S(I))) 
\end{eqnarray}

where the equality holds by definition of $r$. We can also lower bound $f(S(I))$ as follows.  

\begin{eqnarray}\label{eq:SI}
f(S(I)) &=& \sum_{j \in S(I)} \Delta(j, S^{j-1} \cap S(I)) \geq
 \sum_{j \in S(I)} \Delta(j, S^{j-1})  \nonumber \\
 &=& f(\opt')\sum_{j \in S(I)} a_j+b_j+c_j = f(\opt')\sum_{i' \in I} \sum_{\ell \in B_{i'}} a'_{\ell} + b'_{\ell} + c'_{\ell} 
 \end{eqnarray}

where the inequality holds by submodularity of $f$, and equalities hold by definition. By combining inequalities~\ref{eq:S'I}, \ref{eq:Sk2I}, and \ref{eq:SI}, we conclude that:

$$
\frac{f(S_{k_2,I})}{f(\opt')} \geq 
\sum_{i' \in I} \sum_{\ell \in B_{i'}} \left(a'_{\ell} + b'_{\ell} + c'_{\ell}\right)
 + \left(1-e^{-1 - \frac{(4 - \vert I \vert)}{4} \times \frac{r}{1-r}}\right) \left(1-\alpha - \sum_{i' \in I} \sum_{\ell \in B_{i'}}  b'_{\ell} \right)
$$

This proves the feasibility of solution $(\beta, \alpha, \{a'_{\ell}, b'_{\ell}, c'_{\ell}\}_{\ell=1}^{256})$, and consequently completes the proof of Theorem~\ref{thm:54approximate}.
\end{proof}

\section{Small-size Core-sets for Submodular Maximization}\label{sec:small} 
\subsection{Hardness Results for Small-size Core-sets}\label{subsec:hardness}
We start by presenting  the hardness results for non-randomized core-sets. 

\begin{theorem}\label{thm:hardness_submod}
For any $1 \leq k' \leq k$, the output of any algorithm is at most an $O(\frac{k'}{k})$-approximate composable non-randomized core-set for the submodular maximization problem. 
\end{theorem}

\begin{proof}
The hardness instance consists of set $A$ of $k$, and set $B$ of the remaining items. We define the submodular function $f: 2^{A \cup B} \rightarrow \RR$ as $f(S) = |S \cap A|$ for any subset $S \subseteq A \cup B$. 
Suppose the (non-randomized) partitioning, puts all items in $A$ in one part, and distributes the rest of the items in other parts arbitrarily. Since at most $k'$ items in each part can be selected as part of the core-set, no matter what sets are chosen in each part, any size $k$ subset of the union of output sets cannot have $f$ value more than $k'$. On other hand, set $A$ is a size $k$ subset with $f$ value $k$  which implies the $O(\frac{k'}{k})$ hardness gap.
\end{proof}

Now we construct the following family of instances to achieve hardness results for randomized core-sets in the submodular maximization problem.

\begin{definition}\label{def:hard_instance}
We define instance $I^{k,k'}$ of the submodular maximization problem as follows. 
Define $\Gamma$ to be $\lfloor \sqrt{k'k} \rfloor$.
For each $1 \leq i \leq k-\Gamma$, we add an item that represents the set $\{i\}$. 
For each $k-\Gamma < i \leq k$, we add $\lfloor \frac{k}{k'} \rfloor$ identical items all representing the same set $\{i\}$. The value of a subset $S$ of items, $f(S)$, is equal to the cardinality of the union of all sets the items in $S$ represent. This is a coverage valuation function and subsequently monotone submodular.    
\end{definition}

\begin{theorem}\label{thm:hardness_submod-random}
For any $1 \leq k' \leq k$, with $m = \Theta(\frac{k}{k'})$ machines, the output of any algorithm is at most an $O(\sqrt{\frac{k'}{k}})$-approximate randomized composable core-set for the submodular maximization problem. 
\end{theorem}

\begin{proof}
We say item $\{i\}$ is {\em alone} in machine $\ell$ if there is no other item with the same set $\{i\}$ in machine $\ell$ for any $1 \leq i \leq k$, and $1 \leq \ell \leq m$. 
Each item is sent to one of the $m$ machines uniformly at random. We also know that  there are $\Theta(m)$ copies of set $\{i\}$ for each $k - \Gamma < i \leq k$. 
Therefore $Pr[\mbox{Some item with set \{i\} is sent to machine $\ell$, and is alone in machine $\ell$}]$ is $\Omega(1)$ for any  pair of $k-\Gamma < i \leq k$ and $1 \leq \ell \leq m$. 
So machine $\ell$ receives in expectation $\Omega(\Gamma)$ alone items for any $1 \leq \ell \leq m$. 
Since we can permute the elements $\{1, 2, \cdots, k\}$ arbitrarily, there is no difference between alone items with $i \leq k-\Gamma $, and the ones with $i > k-\Gamma$.  We also know that each machine can output at most $k'$ items. So the probability that an arbitrary set $\{i\}$ for some $1 \leq i \leq k-\Gamma$ is selected in its machine is at most $O(\frac{k'}{\Gamma})$. We conclude that the union of all selected sets have size at most $\Gamma + O(\frac{k'}{\Gamma})(k-\Gamma)$ which is $O(\sqrt{\frac{k'}{k}} k)$ by definition of $\Gamma$. It is easy to see that the optimum solution has value $k$ as there exist $k$ disjoint  sets in this instance. So the approximation factor of any algorithm for instance $I^{k,k'}$ is $O(\sqrt{\frac{k'}{k}})$. 
\end{proof}

\subsection{$\Omega(\sqrt{\frac{k'}{k}})$-approximate Core-sets for Monotone Submodular Maximization}\label{subset:small_size_submod}
In Theorem~\ref{thm:hardness_submod-random}, we proved that it is not possible to achieve better than $
O(\sqrt{\frac{k'}{k}})$-approximate randomized composable core-sets for the submodular maximization problem. Here we show that this bound is tight. 
We prove this by applying a randomized algorithm (which is different from algorithm $\greedy$):
\begin{enumerate}
\item Let set $T_{\ell}$ be the set of items sent to machine $1 \leq \ell \leq m$ (each item is sent to one of the machines uniformly at random). 
First run algorithm $\greedy$ on set $T_{\ell}$ to select set $S_{\ell}$ of $k$ items. Note that $S_{\ell}$ has $k$ items, and not $k'$. 

\item  Let $\tau_{\ell}$ to be  $\frac{f(S_{\ell})}{\sqrt{k'k}}$. Now, construct set $S'_{\ell} \subseteq S_{\ell}$ as follows. 
\begin{itemize}
\item Start with $S'_{\ell} = \emptyset$, and while there exists an $x \in S_{\ell}$ with $\Delta(x, S'_{\ell}) \geq \tau_{\ell}$ insert $x$ to $S'_{\ell}$. 
\end{itemize}
\item With probability $1/2$ output a random size $k'$ subset of $S_{\ell}$, and with probability $1/2$, return a random size $k'$ subset of $S'_{\ell}$. 
\end{enumerate}

\begin{theorem}
For any $1 \leq k' \leq k$, and $m \geq \frac{k}{k'}$ machines, the union of output sets by the above algorithm form an $\Omega(\sqrt{\frac{k'}{k}})$-approximate randomized core-set for the submodular maximization problem. 
\end{theorem}

\begin{proof}
For any $1 \leq \ell \leq m$, let $S''_{\ell}$ be the output set (with size $k'$) of machine $\ell$. 
First, we prove that $\E[f_k(\cup_{\ell=1}^m S''_{\ell})]$is $\Omega\left(\frac{\sum_{\ell=1}^m f(S_{\ell})}{m}\right)$. 
Let $m' =\lfloor \frac{k}{k'} \rfloor \le m$. 
We choose a set $M'=\{\ell_1, \ell_2, \cdots, \ell_{m'}\}$ of
$m'$ uniformly at random machines. 
Let $U_i$ be the union of output sets of machines $\ell_1, \ell_2, \cdots, \ell_i$ for any $0 \leq i \leq m'$, i.e., $U_i \defeq \cup_{j=1}^i S''_{\ell_j}$. We note that $U_{m'}$ has size $k'm' \leq k' \frac{k}{k'} = k$. Therefore $f(U_{m'})$ is at most $f_k(\cup_{\ell=1}^m S''_{\ell})$.

Following, we show that $\E[f(U_{m'})] \geq \frac{\sum_{\ell=1}^m f(S_{\ell})}{m}$.  For any $i$, we prove that $\E[f(U_i) - f(U_{i-1})]$ is at least $\frac{(f(S_{\ell_i}) - f(U_{i-1}))k'}{2k}$ where $1 \leq i \leq m'$. We know that adding the whole set $S_{\ell_i}$ to $U_{i-1}$ has marginal value at least $f(S_{\ell_i})-f(U_{i-1})$. With probability $1/2$, output set $S''_{\ell_i}$ consists of $k'$ random items out of the $k$ items of $S_{\ell_i}$. By submodularity, this implies that in expectation, adding set $S''_{\ell_i}$ has marginal value at least $\frac{k'(f(S_{\ell_i})-f(U_{i-1}))}{2k}$ to set $U_{i-1}$. Summing up all these inequalities, and using monotonicity of $f$, we have that: 

$$
\E[f(U_{m'})] \geq \sum_{i=1}^{m'} \E [f(U_i) - f(U_{i-1})] \geq \sum_{i=1}^{m'} \frac{k' \E[f(S_{\ell_i})-f(U_{i-1})]}{2k} \geq
\sum_{i=1}^{m'} \frac{k'(f(S_{\ell_i})-\E[f(U_{m'})])}{2k}
$$

We conclude that $\E[f(U_{m'})]$ is $O\left(\frac{\sum_{i=1}^{m'} f(S_{\ell_i})}{m'}\right)$. 
We also note that  $\E[\frac{\sum_{i=1}^{m'} f(S_{\ell_i})}{m'}]$ is equal to $\frac{\sum_{\ell=1}^m f(S_{\ell})}{m}$ because the expected value of average of some numbers chosen randomly from a set is equal to the average of all numbers in the set. Therefore $\E[f(U_{m'})]$ is $O\left(\frac{\sum_{\ell=1}^{m} f(S_{\ell})}{m}\right)$. 
If $\frac{\sum_{\ell=1}^{m} f(S_{\ell})}{m}$ is 
$\Omega\left(\sqrt{\frac{k'}{k}}f(\opt)\right)$, 
we have $\E[f(U_{m'})] \geq \Omega\left(\sqrt{\frac{k'}{k}}f(\opt)\right)$ which concludes the proof of this theorem. 
So we focus on the case $\frac{\sum_{\ell=1}^{m} f(S_{\ell})}{m} = o\left(\sqrt{\frac{k'}{k}}f(\opt)\right)$, and finish the proof as follows. 

In each machine $i$ for the sake of the analysis,  we select $\lfloor \frac{2k}{m} \rfloor$ items of $\opt \cap T_i$ randomly (if there are not these many items in $\opt \cap T_i$ take all of them) where $1 \leq i \leq m$. Define $\opt'' \subseteq \opt$  be the union of these selected items. 
Since $\{T_i\}_{i=1}^m$ is a random partitioning of all items, we expect $\frac{k}{m}$ items of $\opt$ to be present in set $T_i$ for each $i$. 
For each item $x \in \opt$, with $\Omega(1)$ probability, set $T_{i^x}$ does not have more than  $\lfloor \frac{2k}{m} \rfloor$ items of $\opt$ where $i^x$ is the index of the machine that receives $x$, i.e., $x \in T_{i^x}$. Therefore, any item $x \in \opt$ is in set $\opt''$ with  $\Omega(1)$ probability. We conclude that  $\E[f(\opt'')]$ is $\Omega(f(\opt))$. Using the proof of Lemma~\ref{lem:linear_in_Deltas}, we have that:

$$
f(\opt'' \cap (\cup_{\ell=1}^m S'_{\ell})) \geq 
f(\opt'') -  \sum_{\ell=1}^m \sum_{x \in \opt'' \cap T_{\ell} \setminus S'_{\ell}} \Delta(x,\opt^x) 
$$

Similar to the proofs of Lemmas~\ref{lem:upperbound_not_selected}, and \ref{lem:randomness}, we also have that:
 
 \begin{eqnarray*}
 \sum_{\ell=1}^m \sum_{x \in \opt'' \cap T_{\ell} \setminus S'_{\ell}} \Delta(x,\opt^x) 
 &\leq&
  \sum_{\ell=1}^m | \opt'' \cap T_{\ell} \setminus S'_{\ell}|\tau_{\ell} 
 + \frac{\sum_{\ell=1}^m f(S'_{\ell})}{m} \\
 &\leq&
\frac{\sum_{\ell=1}^m 2k\tau_{\ell} 
 + \sum_{\ell=1}^m f(S'_{\ell})}{m} 
 \leq
\frac{\sqrt{\frac{k}{k'}} \sum_{\ell=1}^m  f(S_{\ell}) + \sum_{\ell=1}^m  f(S'_{\ell})}{m} \\
&\leq&
\frac{(\sqrt{\frac{k}{k'}}+1)\sum_{\ell=1}^m  f(S_{\ell})}{m} 
 \end{eqnarray*}

where the last inequality is implied by the fact that $S'_{\ell} \subseteq S_{\ell}$, and consequently $f(S'_{\ell}) \leq f(S_{\ell})$. We conclude that: 

 \begin{eqnarray*}
 f(\opt'' \cap (\cup_{\ell=1}^m S'_{\ell})) &\geq& 
f(\opt'') - (\sqrt{\frac{k}{k'}}+1)\sum_{\ell=1}^m  f(S_{\ell})/m \\
&=& \Omega(f(\opt)) - \sqrt{\frac{k}{k'}} \times o\left(\sqrt{\frac{k'}{k}}f(\opt)\right) 
= \Omega(f(\opt))
\end{eqnarray*}

where the equations  are implied by the facts that  $\sum_{\ell=1}^m  \frac{f(S_{\ell})}{m}$ is $o(\sqrt{\frac{k'}{k}}f(\opt))$, and $\E[f(\opt'')]$ is $\Omega(f(\opt))$. We note that items of set $\opt'' \cap (\cup_{\ell=1}^m S'_{\ell})$ are not necessarily present in the  output sets $\{S''_{\ell}\}_{\ell=1}^m$. To finish the proof of this theorem, it suffices to prove that each item in $\opt'' \cap S'_{\ell}$ is present in the output set $S''_{\ell}$ with probability $\Omega\left(\sqrt{\frac{k'}{k}}\right)$ for any $1 \leq  \ell \leq m$. 
Each item in $S'_{\ell}$ had marginal value at least $\tau_{\ell}$ to set $S'_{\ell}$ when it was selected to be added to $S'_{\ell}$. 
So there are at most $\frac{f(S'_{\ell})}{\tau_{\ell}} \leq \frac{f(S_{\ell})}{\tau_{\ell}} = \sqrt{k'k}$ items in $S'_{\ell}$. Since with probability $1/2$ the algorithm in machine $\ell$ outputs $k'$ random items of $S'_{\ell}$, each item in $S'_{\ell}$  is in the output set with probability at least $\frac{k'}{2\sqrt{k'k}} = \Omega\left(\sqrt{\frac{k'}{k}}\right)$. Therefore the intersection of $\opt''$ and the union of output sets, $\opt'' \cap (\cup_{\ell=1}^m S''_{\ell})$, has expected value $\Omega\left(\sqrt{\frac{k'}{k}} f(\opt'' \cap (\cup_{\ell=1}^m S'_{\ell}))\right) = \Omega\left(\sqrt{\frac{k'}{k}} f(\opt)\right)$. We note that $\vert \opt'' \cap (\cup_{\ell=1}^m S''_{\ell}) \vert$ is at most $\vert \opt'' \vert \leq \vert \opt \vert = k$. Therefore we have $\E[f_k(\cup_{\ell=1}^m S''_{\ell})] \geq \E[f(\opt'' \cap (\cup_{\ell=1}^m S''_{\ell}))] = \Omega\left(\sqrt{\frac{k'}{k}} f(\opt)\right)$ which concludes the proof of this theorem. 
\end{proof}

\section{Conclusion}
The concept of composable core-sets has been introduced recently in the context of distributed and streaming algorithms and have been applied to several problems~\cite{IMMM14,BBLM14}. In this paper, we introduced the concept of {\em randomized} composable core-sets and showed its effectiveness in maximizing submodular functions in a distributed manner.  
While we mainly discuss the cardinality constraint in this paper, we expect that the ideas and the proof techniques be applicable to more general packing constraints such as matroid constraints.  There are several research problems that are interesting to explore in this line of research.

\begin{itemize}
\item For the submodular maximization problem, it remains to find a randomized composable core-set of approximation factor $1-{1\over e}$, or rule out the possibility of constructing such a core-set.
\item We discussed how the size and multiplicity of the composable core-set can help improve the approximation factor of the algorithms. It would be nice to get tight bounds on the approximation factor for each range of the size  of composable core-sets. Moreover, it would be interesting to study the impact of increasing the multiplicity of the core-set on the achievable approximation factors.
\item While we provided a tight result for the small-size composable core-set problem, the achievable approximation factor is not satisfactory. A natural way to improve this factor is to apply the composable core-set idea iteratively, and achieve an improve approximation factor. Even for composable core-sets of size $k$ and above, it might be possible to improve the approximation factor by applying such a core-set multiple times. This approach leads to several interesting follow-up questions.
\item  While randomized composable core-sets are applicable to random-order streaming models, applying the proof techniques in this paper may result in stronger approximation factors in pure random-order streaming models (compared to the ones presented here). We leave this problem to future research.
\end{itemize}

Finally, it would be nice to explore applicability of these ideas on other optimization and graph theoretic problems.

\bibliographystyle{plain}
\bibliography{refs}

\begin{thebibliography}{10}

\bibitem{AAIMV13}
Sofiane Abbar, Sihem Amer{-}Yahia, Piotr Indyk, Sepideh Mahabadi, and
  Kasturi~R. Varadarajan.
\newblock Diverse near neighbor problem.
\newblock In {\em Symposuim on Computational Geometry 2013, SoCG '13, Rio de
  Janeiro, Brazil, June 17-20, 2013}, pages 207--214, 2013.

\bibitem{agarwal2012mergeable}
Pankaj~K Agarwal, Graham Cormode, Zengfeng Huang, Jeff Phillips, Zhewei Wei,
  and Ke~Yi.
\newblock Mergeable summaries.
\newblock In {\em Proceedings of the 31st symposium on Principles of Database
  Systems}, pages 23--34. ACM, 2012.

\bibitem{agarwal2004approximating}
Pankaj~K Agarwal, Sariel Har-Peled, and Kasturi~R Varadarajan.
\newblock Approximating extent measures of points.
\newblock {\em Journal of the ACM (JACM)}, 51(4):606--635, 2004.

\bibitem{AlonSpencerBook2000}
Noga Alon and Joel Spencer.
\newblock {\em The Probabilistic Method}.
\newblock John Wiley, 2000.

\bibitem{ANOY14}
Alexandr Andoni, Aleksandar Nikolov, Krzysztof Onak, and Grigory Yaroslavtsev.
\newblock Parallel algorithms for geometric graph problems.
\newblock In {\em Symposium on Theory of Computing, {STOC} 2014, New York, NY,
  USA, May 31 - June 03, 2014}, pages 574--583, 2014.

\bibitem{AK14}
Ashwinkumar Badanidiyuru and Amin Karbasi.
\newblock Personal communication, 2014.

\bibitem{karbasiKDD2014}
Ashwinkumar Badanidiyuru, Baharan Mirzasoleiman, Amin Karbasi, and Andreas
  Krause.
\newblock Streaming submodular maximization: Massive data summarization on the
  fly.
\newblock In {\em KDD}, 2014.

\bibitem{AshwinVondrakSODA2014}
Ashwinkumar Badanidiyuru and Jan Vondrák.
\newblock Fast algorithms for maximizing submodular functions.
\newblock In {\em SODA}, pages 1497--1514, 2014.

\bibitem{nips-BEL13}
Maria-Florina Balcan, Steven Ehrlich, and Yingyu Liang.
\newblock Distributed clustering on graphs.
\newblock In {\em NIPS}, page to appear, 2013.

\bibitem{BBLM14}
MohammadHossein Bateni, Aditya Bhashkara, Silvio Lattanzi, and Vahab Mirrokni.
\newblock Mapping core-sets for balanced clustering.
\newblock In {\em Advances in Neural Information Processing Systems 26: 27th
  Annual Conference on Neural Information Processing Systems 2013. Proceedings
  of a meeting held December 5-8, 2013, Lake Tahoe, Nevada, United States.},
  2014.

\bibitem{BST12}
Guy~E. Blelloch, Harsha~Vardhan Simhadri, and Kanat Tangwongsan.
\newblock Parallel and i/o efficient set covering algorithms.
\newblock In {\em SPAA}, pages 82--90, 2012.

\bibitem{BuchbinderFeldmanNaorSchwartzSODA2014}
Niv Buchbinder, Moran Feldman, Joseph~(Seffi) Naor, and Roy Schwartz.
\newblock Submodular maximization with cardinality constraints.
\newblock In {\em Proceedings of the Twenty-Fifth Annual ACM-SIAM Symposium on
  Discrete Algorithms}, SODA '14, pages 1433--1452. SIAM, 2014.

\bibitem{CKT10}
Flavio Chierichetti, Ravi Kumar, and Andrew Tomkins.
\newblock Max-cover in map-reduce.
\newblock In {\em WWW}, pages 231--240, 2010.

\bibitem{CW13}
Kenneth~L. Clarkson and David~P. Woodruff.
\newblock Low rank approximation and regression in input sparsity time.
\newblock In {\em Symposium on Theory of Computing Conference, STOC'13, Palo
  Alto, CA, USA, June 1-4, 2013}, pages 81--90, 2013.

\bibitem{CKW10}
Graham Cormode, Howard~J. Karloff, and Anthony Wirth.
\newblock Set cover algorithms for very large datasets.
\newblock In {\em CIKM}, pages 479--488, 2010.

\bibitem{osdi-DG04}
Jeffrey Dean and Sanjay Ghemawat.
\newblock Mapreduce: Simplified data processing on large clusters.
\newblock In {\em OSDI}, pages 137--150, 2004.

\bibitem{FMSW10}
Dan Feldman, Morteza Monemizadeh, Christian Sohler, and David~P. Woodruff.
\newblock Coresets and sketches for high dimensional subspace approximation
  problems.
\newblock In {\em Proceedings of the Twenty-First Annual {ACM-SIAM} Symposium
  on Discrete Algorithms, {SODA} 2010, Austin, Texas, USA, January 17-19,
  2010}, pages 630--649, 2010.

\bibitem{GMMO}
S.~Guha, N.~Mishra, R.~Motwani, and L.~O'Callaghan.
\newblock Clustering data streams.
\newblock {\em STOC}, 2001.

\bibitem{IMMM14}
Piotr Indyk, Sepideh Mahabadi, Mohammad Mahdian, and Vahab Mirrokni.
\newblock Composable core-sets for diversity and coverage maximization.
\newblock In {\em ACM PODS}, 2014.

\bibitem{KKS14}
Michael Kapralov, Sanjeev Khanna, and Madhu Sudan.
\newblock Approximating matching size from random streams.
\newblock In {\em SODA}, pages 734--751, 2014.

\bibitem{soda-KSV10}
Howard~J. Karloff, Siddharth Suri, and Sergei Vassilvitskii.
\newblock A model of computation for mapreduce.
\newblock In {\em SODA}, pages 938--948, 2010.

\bibitem{socc14}
Raimondas kiveris, Silvio Lattanzi, Vahab Mirrokni, Vibhor Rastogi, and Sergei
  Vasilvitski.
\newblock Connected components in mapreduce and beyond.
\newblock In {\em ACM SOCC}, 2014.

\bibitem{KMVV13}
Ravi Kumar, Benjamin Moseley, Sergei Vassilvitskii, and Andrea Vattani.
\newblock Fast greedy algorithms in mapreduce and streaming.
\newblock In {\em SPAA}, pages 1--10, 2013.

\bibitem{spaa-LMSV11}
Silvio Lattanzi, Benjamin Moseley, Siddharth Suri, and Sergei Vassilvitskii.
\newblock Filtering: a method for solving graph problems in mapreduce.
\newblock In {\em SPAA}, pages 85--94, 2011.

\bibitem{nips13}
Baharan Mirzasoleiman, Amin Karbasi, Rik Sarkar, and Andreas Krause.
\newblock Distributed submodular maximization: Identifying representative
  elements in massive data.
\newblock In {\em Advances in Neural Information Processing Systems 26: 27th
  Annual Conference on Neural Information Processing Systems 2013. Proceedings
  of a meeting held December 5-8, 2013, Lake Tahoe, Nevada, United States.},
  pages 2049--2057, 2013.

\bibitem{Nemhauser_Wolsey_Fisher78}
G.~L. Nemhauser, L.~A. Wolsey, and M.~L. Fisher.
\newblock An analysis of approximations for maximizing submodular set
  functions.
\newblock {\em Mathematical Programming}, 14(1):265--294, 1978.

\bibitem{Sarlos06}
Tam{\'{a}}s Sarl{\'{o}}s.
\newblock Improved approximation algorithms for large matrices via random
  projections.
\newblock In {\em 47th Annual {IEEE} Symposium on Foundations of Computer
  Science {(FOCS} 2006), 21-24 October 2006, Berkeley, California, USA,
  Proceedings}, pages 143--152, 2006.

\end{thebibliography}

\end{document}